\documentclass[a4paper,10pt]{article}
\usepackage{amsthm,amsmath,amssymb}
\usepackage{subfig,sidecap,caption}
\usepackage{paralist}
\usepackage[usenames,dvipsnames]{color}
\usepackage[pdftex,breaklinks,colorlinks,
    citecolor={BlueViolet}, linkcolor={Blue},urlcolor=Maroon]{hyperref}  
\usepackage{tkz-graph}
\usepackage{graphicx}
\usetikzlibrary{decorations.pathmorphing}
\usepackage{charter,eulervm}%
\usepackage{tabularx}
\usepackage[final,expansion=alltext,protrusion=true]{microtype}
\graphicspath{{./figs/}}

\theoremstyle{plain} 
\newtheorem{theorem}{Theorem}[section]
\newtheorem{lemma}[theorem]{Lemma}
\newtheorem{corollary}[theorem]{Corollary}
\newtheorem{proposition}[theorem]{Proposition}

\newtheorem{claim}{Claim}
\def\boxit#1{\vbox{\hrule\hbox{\vrule\kern4pt
      \vbox{\kern1pt#1\kern1pt} \kern2pt\vrule}\hrule}}

\newcommand{\uivd}{{unit interval vertex deletion}}
\newcommand{\uied}{{unit interval edge deletion}}
\newcommand{\uic}{{unit interval completion}}
\newcommand{\uie}{{unit interval editing}}
\newcommand{\nhcag}{normal Helly circular-arc graph}
\newcommand{\phcag}{proper Helly circular-arc graph}
\newcommand{\lp}[1]{\ensuremath{{\mathtt{lp}(#1)}}}
\newcommand{\rp}[1]{\ensuremath{{\mathtt{rp}(#1)}}}
\newcommand{\lpp}[1]{\ensuremath{{\mathtt{lp'}(#1)}}}
\newcommand{\rpp}[1]{\ensuremath{{\mathtt{rp'}(#1)}}}
\newcommand{\ccp}[1]{\ensuremath{{\mathtt{ccp}(#1)}}}
\newcommand{\cp}[1]{\ensuremath{{\mathtt{cp}(#1)}}}
\newcommand{\comment}[1]{\textbackslash\!\!\textbackslash {\em #1}}

\title{Unit Interval Editing is Fixed-Parameter Tractable}

\author{Yixin Cao\thanks{Department of Computing, Hong Kong
    Polytechnic University, Hong Kong, China.
    \href{mailto:yixin.cao@polyu.edu.hk}{\tt yixin.cao@polyu.edu.hk}.
    Supported in part by the Hong Kong Research Grants Council (RGC)
    under grant 252026/15E, the National Natural Science
    Foundation of China (NSFC) under grants 61572414 and 61420106009,
    the Hong Kong Polytechnic University (PolyU) under grant 4-ZZEZ,
    and the European Research Council (ERC) under grant 280152.}}

\date{}

\begin{document}
\maketitle
\begin{abstract}
  Given a graph~$G$ and integers $k_1$, $k_2$, and~$k_3$, the unit interval editing problem asks whether $G$ can be transformed into a unit interval graph by at most $k_1$ vertex deletions, $k_2$ edge deletions, and $k_3$ edge additions.  We give an algorithm solving this problem in time $2^{O(k\log k)}\cdot (n+m)$, where $k := k_1 + k_2 + k_3$, and $n, m$ denote respectively the numbers of vertices and edges of $G$.  Therefore, it is fixed-parameter tractable parameterized by the total number of allowed operations.  

Our algorithm implies the fixed-parameter tractability of the unit interval edge deletion problem, for which we also present a more efficient algorithm running in time $O(4^k \cdot (n + m))$.  Another result is an $O(6^k \cdot (n + m))$-time algorithm for the unit interval vertex deletion problem, significantly improving the algorithm of van 't Hof and Villanger, which runs in time $O(6^k \cdot n^6)$.
\end{abstract}

\section{Introduction}\label{sec:intro}
A graph is a \emph{unit interval graph} if its vertices can be
assigned to unit-length intervals on the real line such that there is
an edge between two vertices if and only if their corresponding
intervals intersect.  Most important applications of unit interval
graphs were found in computational biology
\cite{bodlaender-96-intervalize-k-colored-graphs,goldberg-95-interval-edge-deletion,kaplan-99-chordal-completion},
where data are mainly obtained by unreliable experimental methods.
Therefore, the graph representing the raw data is very unlikely to be
a unit interval graph, and an important step before understanding the
data is to find and fix the hidden errors.  For this purpose various
graph modification problems have been formulated: Given a graph $G$ on
$n$ vertices and $m$ edges, is there a set of at most $k$
modifications that make $G$ a unit interval graph.  In particular,
edge additions, also called completion, and edge deletions are used to fix false
negatives and false positives respectively, while vertex deletions can
be viewed as the elimination of outliers.  We have thus three
variants, which are all known to be NP-complete
\cite{lewis-80-node-deletion-np, 
  yannakakis-81-minimum-fill-in, goldberg-95-interval-edge-deletion}.

These modification problems to unit interval graphs have been well
studied in the framework of parameterized computation, where the
parameter is usually the number of modifications.  Recall that a graph
problem, with a nonnegative parameter~$k$, is {\em fixed-parameter
  tractable (FPT)} if there is an algorithm solving it in time
$f(k)\cdot (n + m)^{O(1)}$, where $f$ is a computable function
depending only on~$k$ \cite{downey-13}.  The problems \uic{}
and \uivd{} have been shown to be FPT by Kaplan et
al.~\cite{kaplan-99-chordal-completion} and van Bevern et
al.~\cite{bevern-10-pivd} respectively.  In contrast, however, the
parameterized complexity of the edge deletion version remained open to
date, which we settle here.  Indeed, we devise single-exponential
linear-time parameterized algorithms for both deletion versions.
\begin{theorem}\label{thm:alg-1}
  The problems \uivd{} and \uied{} can be solved in time $O(6^k\cdot
  (n + m))$ and $O(4^k\cdot (n + m))$ respectively.
\end{theorem}
Our algorithm for \uivd{} significantly improves the currently
best parameterized algorithm for it, which takes $O(6^k\cdot
n^6)$ time \cite{villanger-13-pivd}.  Another algorithmic result of
van 't Hof and Villanger \cite{villanger-13-pivd} is an $O(n^7)$-time
$6$-approximation algorithm for the problem, which we improve to the
following.
\begin{theorem}\label{thm:alg-2}
  There is an $O(n m + n^2)$-time approximation algorithm of
  approximation ratio~$6$ for the minimization version of the \uivd{}
  problem.
\end{theorem}

\tikzstyle{vertex}  = [{fill=blue,circle,blue,draw,inner sep=1pt}]
\begin{figure*}[t]
  \centering\small
  \subfloat[claw]{\label{fig:claw}
    \includegraphics{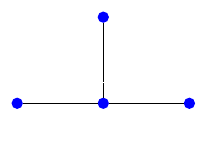} 
  }
  \,
  \subfloat[$S_3$]{\label{fig:tent}
    \includegraphics{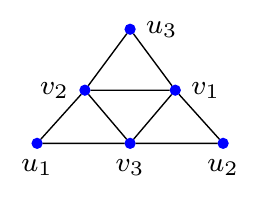} 
  }
  \,
  \subfloat[$\overline{S_3}$]{\label{fig:net}
    \includegraphics{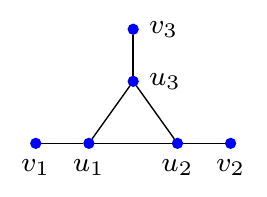}
  }
  \quad
  \subfloat[{$W_5$}]{\label{fig:5-wheel}
    \includegraphics{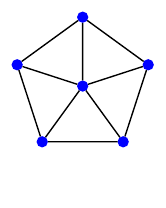} 
  }
  \quad
  \subfloat[{$C^*_5$} = $\overline{W_5}$]{\label{fig:5-star}
    \includegraphics{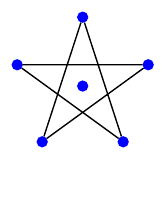} 
  }
  \quad
  \subfloat[$\overline{C_6}$]{\label{fig:complement-c-6}
    \includegraphics{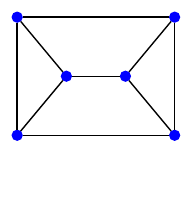} 
  }
  \caption{Small forbidden induced graphs.}
  \label{fig:fis}
\end{figure*}
The structures and recognition of unit interval graphs have been well
studied and well understood \cite{deng-96-proper-interval-and-cag}.
It is known that a graph is a unit interval graph if and only if it
contains no claw, $S_3, \overline{S_3}$, (as depicted in
Fig.~\ref{fig:fis},) or any hole (i.e., an induced cycle on at least
four vertices) \cite{roberts-69-indifference-graphs,
  wegner-67-dissertation}.  Unit interval graphs are thus a subclass
of chordal graphs, which are those graphs containing no holes.
Modification problems to chordal graphs and unit interval graphs are
among the earliest studied problems in parameterized computation, and
their study had been closely related.  For example, the algorithm of
Kaplan et al.~\cite{kaplan-99-chordal-completion} for \uic{} is a
natural spin-off of their algorithm for chordal completion, or more
specifically, the combinatorial result of all minimal ways to fill
holes in.  A better analysis was shortly done by
Cai~\cite{cai-96-hereditary-graph-modification}, who also made
explicit the use of bounded-depth search in disposing of finite
forbidden induced subgraphs.  This observation and the parameterized
algorithm of Marx \cite{marx-10-chordal-deletion} for the chordal
vertex deletion problem immediately imply the fixed-parameter
tractability of the \uivd\ problem: One may break first all induced
claws, $S_3$'s, and $\overline{S_3}$'s, and then call Marx's
algorithm.  Here we are using the hereditary property of unit interval
graphs,---recall that a graph class is {\em hereditary} if it is
closed under taking induced subgraphs.  However, neither approach can
be adapted to the edge deletion version in a simple way.  Compared to
completion that needs to add $\Omega(\ell)$ edges to fill a $C_\ell$
(i.e., a hole of length $\ell$) in, an arbitrarily large hole can be
fixed by a single edge deletion.  On the other hand, the deletion of
vertices leaves an induced subgraph, which allows us to focus on holes
once all claws, $S_3$'s, and $\overline{S_3}$'s have been eliminated;
however, the deletion of edges to fix holes of a \{claw, $S_3,
\overline{S_3}$\}-free graph may introduce new claws, $S_3$'s, and/or
$\overline{S_3}$'s.  Therefore, although a parameterized algorithm for
the chordal edge deletion problem has also been presented by Marx
\cite{marx-10-chordal-deletion}, there is no obvious way to use it to
solve the \uied\ problem.

Direct algorithms for \uivd\ were later discovered by van Bevern et
al.~\cite{bevern-10-pivd} and van 't Hof and
Villanger~\cite{villanger-13-pivd}, both using a two-phase approach.
The first phase of their algorithms breaks all forbidden induced
subgraphs on at most six vertices.  Note that this differentiates from
the aforementioned simple approach in that it breaks not only claws,
$S_3$'s, and $\overline{S_3}$'s, but all $C_\ell$'s with $\ell \le 6$.
Although this phase is conceptually intuitive, it is rather nontrivial
to efficiently carry it out, and the simple brute-force way introduces
an $n^6$ factor to the running time.  Their approaches diverse
completely in the second phase.  Van Bevern et
al.~\cite{bevern-10-pivd} used a complicated iterative compression
procedure that has a very high time complexity, while van 't Hof and
Villanger~\cite{villanger-13-pivd} showed that after the first phase,
the problem is linear-time solvable.  The main observation of van 't
Hof and Villanger~\cite{villanger-13-pivd} is that a connected \{claw,
$S_3, \overline{S_3}$, $C_4$, $C_5$, $C_6$\}-free graph is a proper
circular-arc graph, whose definition is postponed to
Section~\ref{sec:structures}.  In the conference presentation where
Villanger first announced the result, it was claimed that this settles
the edge deletion version as well.
However, the claimed result has not been materialized: It appears
neither in the conference version \cite{villanger-10-pivd} (which has
a single author) nor in the significantly revised and extended journal
version \cite{villanger-13-pivd}.  Unfortunately, this unsubstantiated
claim did get circulated.

Although the algorithm of van 't Hof and
Villanger~\cite{villanger-13-pivd} is nice and simple, its
self-contained proof is excruciatingly complex.  We revisit the
relation between unit interval graphs and some subclasses of proper
circular-arc graphs, and study it in a structured way.  In particular,
we observe that unit interval graphs are precisely those graphs that
are both chordal graphs and proper Helly circular-arc graphs.  As a
matter of fact, unit interval graphs can also be viewed as ``unit
Helly interval graphs'' or ``proper Helly interval graphs,''
thereby making a natural subclass of \phcag{s}.  The full containment
relations are summarized in Fig.~\ref{fig:classes-containment-2}; the
reader unfamiliar with some graph classes in this figure may turn to
the appendix for a brief overview.  These observations inspire us to
show that a connected \{claw, $S_3, \overline{S_3}$, $C_4$,
$C_5$\}-free graph is a proper Helly circular-arc graph.  It is easy
to adapt the linear-time certifying recognition algorithms for
\phcag{s} \cite{lin-13-nhcag-and-subclasses, cao-15-nhcag} to detect
an induced claw, $S_3, \overline{S_3}, C_4$, or $C_5$ if one exists.
Once all of them have been completely eliminated and the graph is a
\phcag{}, it is easy to solve the \uivd{} problem in linear time.
Likewise, using the structural properties of \phcag{s}, we can derive
a linear-time algorithm for \uied{} on them.  It is then
straightforward to use simple branching to develop the parameterized
algorithms stated in Theorem~\ref{thm:alg-1}, though some nontrivial
analysis is required to obtain the time bound for \uied{}.
\begin{figure}[t]
  \centering\small
  \captionsetup{justification=centering}
    \includegraphics{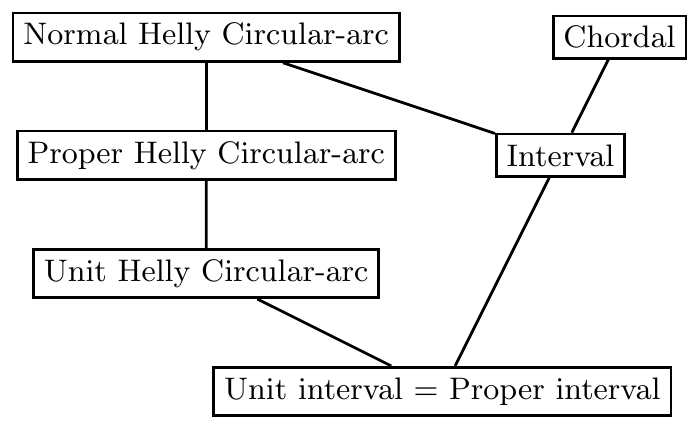} 
  \caption{Containment relation of related graph classes.\\ 
    Normal Helly circular-arc $\cap$ Chordal = Interval. \\Proper Helly
    circular-arc $\cap$ Chordal = Unit Helly circular-arc $\cap$
    Chordal = Unit interval.}
  \label{fig:classes-containment-2}
\end{figure}

Van Bevern et al.~\cite{bevern-10-pivd} showed that the \uivd{}
problem remains NP-hard on \{claw, $S_3, \overline{S_3}$\}-free
graphs.  After deriving a polynomial-time algorithm for the problem on
\{claw, $S_3, \overline{S_3}, C_4, C_5, C_6$\}-free graphs, van 't Hof
and Villanger~\cite{villanger-13-pivd} asked for its complexity on
\{claw, $S_3, \overline{S_3}, C_4$\}-free graphs.  (It is somewhat
intriguing that they did not mention the \{claw, $S_3, \overline{S_3},
C_4, C_5$\}-free graphs.)  Note that a \{claw, $S_3, \overline{S_3},
C_4$\}-free graph is not necessarily a proper (Helly) circular-arc
graph, evidenced by the $W_5$ (Fig.~\ref{fig:5-wheel}).  We answer
this question by characterizing connected \{claw, $S_3,
\overline{S_3}$, $C_4$\}-free graphs that are not \phcag{s}.  We show
that such a graph must be \textit{like} a $W_5$: If we keep only one
vertex from each twin class (all vertices in a twin class have the
same closed neighborhood) of the original graph, then we obtain a
$W_5$.  It is then routine to solve the problem in linear time.

\begin{theorem}\label{thm:vertex-deletion-phcag}
  The problems \uivd{} and \uied{} can be solved in $O(n + m)$ time on
  \{claw, $S_3, \overline{S_3}$, $C_4$\}-free graphs. 
\end{theorem}

We remark that the techniques we developed in previous work
\cite{cao-16-almost-interval-recognition} can also be used to derive
Theorems~\ref{thm:alg-1} and \ref{thm:alg-2}.  Those techniques,
designed for interval graphs, are nevertheless far more complicated
than necessary when applied to unit interval graphs.  The approach we
used in the current work, i.e., based on structural properties of
\phcag{s}, is tailored for unit interval graphs, hence simpler and
more natural.  Another benefit of this approach is that it enables us
to devise a parameterized algorithm for the general modification
problem to unit interval graphs, which allows all three types of
operations.
This formulation generalizes all the three single-type modifications,
and is also natural from the viewpoint of the aforementioned
applications for de-noising data, where different types of errors are
commonly found coexisting.  Indeed, the assumption that the input data
contain only a single type of errors is somewhat counterintuitive.
Formally, given a graph~$G$, the \emph{unit interval editing } problem
asks whether there are a set~$V_-$ of at most $k_1$ vertices, a
set~$E_-$ of at most~$k_2$ edges, and a set~$E_+$ of at most~$k_3$
non-edges, such that the deletion of $V_-$ and $E_-$ and the addition
of $E_+$ make $G$ a unit interval graph.  We show that it is FPT,
parameterized by the total number of allowed operations, $k := k_1 +
k_2 + k_3$.
\begin{theorem}\label{thm:alg-interval-editing}
  The unit interval editing problem can be solved in time $2^{{O}(k
    \log{k})}\cdot (n + m)$.
\end{theorem}

By and large, our algorithm for unit interval editing again uses the
two-phase approach.  However, we are not able to show that it can be
solved in polynomial time on \phcag{s}.  Therefore, in the first
phase, we use brute force to remove not only claws, $S_3$'s,
$\overline{S_3}$'s, and $C_4$'s, also all holes of length at most $k_3
+ 3$.  The high exponential factor in the running time is due to
purely this phase. After that, every hole has length at least $k_3 +
4$, and has to be fixed by deleting a vertex or edge.  We manage to
show that an inclusion-wise minimal solution of this reduced graph
does not add edges, and the problem can then be solved in linear time.

The study of general modification problems was initiated by Cai
\cite{cai-96-hereditary-graph-modification}, who observed that the
problem is FPT if the objective graph class has a finite number of
minimal forbidden induced subgraphs.  More challenging is thus to
devise parameterized algorithms for those graph classes whose minimal
forbidden induced subgraphs are infinite.  Prior to this paper, the
only known nontrivial graph class on which the general modification
problem is FPT is the chordal graphs \cite{cao-16-chordal-editing}.
Theorem~\ref{thm:alg-interval-editing} extends this territory by
including another well-studied graph class.  As a corollary,
Theorem~\ref{thm:alg-interval-editing} implies the fixed-parameter
tractability of the unit interval edge editing problem, which allows
both edge operations but not vertex deletions
\cite{burzyn-06-NPC-edge-modification}.  To see this we can simply try
every combination of $k_2$ and $k_3$ as long as $k_2+k_3$ does not
exceed the given bound.

\paragraph{Organization.} The rest of the paper is organized as
follows.  Section~\ref{sec:structures} presents combinatorial and
algorithmic results on \{claw, $S_3, \overline{S_3}$, $C_4$\}-free
graphs.  Sections~\ref{sec:vertex-deletion} and
\ref{sec:edge-deletion} present the algorithms for \uivd\ and \uied\
respectively
(Theorems~\ref{thm:alg-1}--\ref{thm:vertex-deletion-phcag}).
Section~\ref{sec:general} extends them to solve the general editing
problem (Theorem~\ref{thm:alg-interval-editing}).
Section~\ref{sec:remarks} closes this paper by discussing some
possible improvement and new directions.  The appendix provides a
brief overview of related graph classes as well as their
characterizations by forbidden induced subgraphs.

\section{\{Claw, $S_3, \overline{S_3}$, $C_4$\}-free graphs}\label{sec:structures}
All graphs discussed in this paper are undirected and simple.  A graph
$G$ is given by its vertex set $V(G)$ and edge set $E(G)$, whose
cardinalities will be denoted by $n$ and $m$ respectively.  All input
graphs in this paper are assumed to be nontrivial ($n > 1$) and
connected, hence $n = O(m)$.
For $\ell\ge 4$, we use $C_\ell$ to denote a hole on $\ell$ vertices;
if we add a new vertex to a $C_\ell$ and make it adjacent to no or all
vertices in the hole, then we end with a $C_\ell^*$ or $W_\ell$,
respectively.  For a hole $H$, we have $|V(H)| = |E(H)|$, denoted by
$|H|$.  The complement graph $\overline G$ of a graph $G$ is defined
on the same vertex set $V(G)$, where a pair of vertices $u$ and $v$ is
adjacent if and only if $u v \not\in E(G)$; e.g., $\overline{W_5} =
C^*_5$, and depicted in Fig.~\ref{fig:complement-c-6} is the
complement of $C_6$.

An \emph{interval graph} is the intersection graph of a set of
intervals on the real line.  A natural way to extend interval graphs
is to use arcs and a circle in the place of intervals and the real
line, and the intersection graph of arcs on a circle is a {\em
  circular-arc graph}.  The set of intervals or arcs is called an {\em
  interval model} or {\em arc model} respectively, and it can be
specified by their $2 n$ endpoints.  In this paper, all intervals and
arcs are closed, and no distinct intervals or arcs are allowed to
share an endpoint in the same model; these restrictions do not
sacrifice any generality.  In a {\em unit interval model} or a {\em
  unit arc model}, every interval or arc has length one.  An interval
or arc model is \emph{proper} if no interval or arc in it properly
contains another interval or arc.  A graph is a {\em unit interval
  graph, proper interval graph, unit circular-arc graph}, or {\em
  proper circular-arc graph} if it has a unit interval model, proper
interval model, unit arc model, or proper arc model, respectively.
The forbidden induced subgraphs of unit interval graphs have been long
known.
\begin{theorem}[\cite{wegner-67-dissertation}]\label{thm:uig-fis}
  A graph is a unit interval graph if and only if it contains no claw,
  $S_3, \overline{S_3}$, or any hole.
\end{theorem}

Clearly, any (unit/proper) interval model can be viewed as a
(unit/proper) arc model leaving some point uncovered, and hence all
(unit/proper) interval graphs are always (unit/proper) circular-arc
graphs.  A unit interval/arc model is necessarily proper, but the
other way does not hold true in general.  A well-known result states
that a proper interval model can always be made unit, and thus these
two graph classes coincide
\cite{roberts-69-indifference-graphs,wegner-67-dissertation}.\footnote{The
  reason we choose ``unit'' over ``proper'' in the title of this paper
  is twofold.  On the one hand, the applications we are interested in
  are more naturally represented by unit intervals.  On the other
  hand, we want to avoid the use of ``proper interval subgraphs,''
  which is ambiguous.}  This fact will be heavily used in the present
paper; e.g., most of our proofs consist in modifying a proper arc
model into a proper interval model, which represents the desired unit
interval graph.  On the other hand, it is easy to check that the $S_3$
is a proper circular-arc graph but not a unit circular-arc graph.
Therefore, the class of unit circular-arc graphs is a proper subclass
of proper circular-arc graphs.  

An arc model is {\em Helly} if every set of pairwise intersecting arcs
has a common intersection.  A {circular-arc graph} is \emph{proper
  Helly} if it has an {arc model} that is both proper and
Helly.

\begin{theorem}[\cite{tucker-74-structures-cag, lin-13-nhcag-and-subclasses, cao-15-nhcag}]\label{thm:phcag}
  A graph is a \phcag{} if and only if it contains no claw, $S_3,
  \overline{S_3}$, $W_4$, $W_5$, $\overline{C_6}$, or $C_\ell^*$ for
  $\ell \ge 4$.
\end{theorem}
The following is immediate from Theorems~\ref{thm:uig-fis}
and~\ref{thm:phcag}.
\begin{corollary}\label{lem:chordal-phcag}
  If a \phcag{} is chordal, then it is a unit interval graph.
\end{corollary}

From Theorems~\ref{thm:uig-fis} and \ref{thm:phcag} one can also
derive the following combinatorial result.  But since we will prove a
stronger result in Theorem~\ref{thm:characterization} that implies it,
we omit its proof here.
\begin{proposition}\label{lem:phcag}
  Every connected \{claw, $S_3, \overline{S_3}$, $C_4$, $C_5$\}-free graph is a
  \phcag{}.
\end{proposition}
Note that in Proposition~\ref{lem:phcag}, as well as most
combinatorial statements to follow, we need the graph to be connected.
Circular-arc graphs are not closed under taking disjoint unions.  If a
(proper Helly) circular-arc graph is not chordal, then it is
necessarily connected.  In other words, a disconnected (proper)
circular-arc graph must be a (unit) interval graph.

Proposition~\ref{lem:phcag} can be turned into an algorithmic
statement.  We say that a recognition algorithm (for a graph class) is
\emph{certifying} if it provides a minimal forbidden induced subgraph
of the input graph $G$ when $G$ is determined to be not in this class.
Linear-time certifying algorithms for recognizing \phcag{s} have been
reported by Lin et al.~\cite{lin-13-nhcag-and-subclasses} and Cao et
al.~\cite{cao-15-nhcag}, from which one can derive a linear-time
algorithm for detecting an induced claw, $\overline{S_3}$, $S_3$,
$C_4$, or $C_5$ from a graph that is not a \phcag{}.

This would suffice for us to develop most of our main results.
Even so, we would take pain to prove slightly stronger results (than
Proposition~\ref{lem:phcag}) on $\cal F$-free graphs, where $\cal F$
denotes the set \{claw, $S_3, \overline{S_3}$, $C_4$\}.  The purpose
is threefold.  First, they enable us to answer the question asked by
van 't Hof and Villanger~\cite{villanger-13-pivd}, i.e., the
complexity of \uivd\ on $\cal F$-free graphs, thereby more accurately
delimiting the complexity border of the problem.  Second, as we will
see, the disposal of $C_5$'s would otherwise dominate the second phase
of our algorithm for \uied{}, so excluding them enables us to obtain
better exponential dependency on $k$ in the running time.  Third, the
combinatorial characterization might be of its own interest.

A \emph{(true) twin class} of a graph $G$ is an inclusion-wise maximal
set of vertices in which all have the same closed neighborhood.  A
graph is called a {\em fat $W_5$} if it has precisely six twin classes
and it becomes a $W_5$ after we remove all but one vertices from each
twin class.\footnote{Interestingly, in the study of proper
  circular-arc graphs that are chordal, Bang-Jensen and
  Hell~\cite{bang-jensen-94-chordal-proper-cag} showed that if a
  connected \{claw, $\overline{S_3}$, $C_4, C_5\}$-free graph contains
  an $S_3$, then it must be a fat $S_3$, which is defined analogously
  as a fat $W_5$.}  By definition, vertices in each twin class induce
a clique.  The five cliques corresponding to hole in the $W_5$ is the
{\em fat hole}, and the other clique is the {\em hub}, of this fat
$W_5$.
%
%

\begin{theorem}\label{thm:characterization}
  Let $G$ be a connected graph.  
  \begin{enumerate}[(1)]
  \item If $G$ is $\cal F$-free, then it is either a fat $W_5$ or a
    \phcag{}.
  \item In $O(m)$ time we can either detect an induced subgraph of $G$
    in $\cal F$, partition $V(G)$ into six cliques constituting a fat
    $W_5$, or build a proper and Helly arc model for $G$.
  \end{enumerate}
\end{theorem}
\begin{figure}[h]
\setbox4=\vbox{\hsize28pc \noindent\strut
\begin{quote}
  \vspace*{-5mm}

  {\bf Algorithm} {\bf recognize-$\cal F$-free}($G$)
  \\
  {\bf Input}: a connected graph $G$.
  \\
  {\bf Output}: a proper and Helly arc model, a subgraph in $\cal F$,
  or six cliques making a fat $W_5$.
  \begin{tabbing}
    AAA\=aaA\=Aaa\=MMMMAAAAAAAAAAAAa\=A \kill
    0.\> call the recognition algorithm for \phcag{s} \cite{cao-15-nhcag};
    \\
    1.\> {\bf if} a proper and Helly arc model $\cal A$ is found {\bf then 
      return} it;
    \\
    2.\> {\bf if} a claw, $\overline{S_3}$, or $S_3$ is found {\bf then return} it;
    \\
    3.\> {\bf if} a $W_4$, $C^*_4$, or $\overline{C_6}$ is found {\bf then 
      return} a $C_4$ contained in it;
    \\
    4.\> {\bf if} a $C^*_\ell$ with hole $H$ and isolated vertex $v$ is
    found {\bf then}
    \\
    4.1.\>\> use breadth-first search to find a shortest path $v \cdots x y h_i$ 
    from $v$ to $H$;
    \\
    4.2.\>\> {\bf if} $y$ has a single neighbor $h_i$ in $H$ {\bf then return} 
    claw $\{h_i, y, h_{i-1}, h_{i+1}\}$; 
    \\
    4.3.\>\> {\bf if} $y$ has only two neighbors on $H$ that are consecutive, 
    say, $\{h_i, h_{i+1}\}$ {\bf then}
    \\
    \>\>\> {\bf return} $\overline{S_3}$ $\{x, y, h_{i-1}, h_{i}, h_{i+1}, h_{i+2}\}$; 
    \\
    4.4.\>\> {\bf return} claw $\{y, x, h_j, h_{j'}\}$, where $h_j, h_{j'}$ are two
    nonadjacent vertices in $N[y]\cap H$;
    \\
    \comment{\;The outcome of step~0 must be a $W_5$; let it be $H$ 
      and $v$.   All subscripts of $h_i$ and $K_i$ are modulo $5$.} 
    \\
    5. \> $K_0\leftarrow\{h_0\}$;$K_1\leftarrow\{h_1\}$;$K_2\leftarrow\{h_2\}$;
    $K_3\leftarrow\{h_3\}$;$K_4\leftarrow\{h_4\}$;$K_v\leftarrow\{v\}$;
    \\
    6.\> {\bf for each} vertex $x$ not in the $W_5$ {\bf do}
    \\
    6.1.\>\> {\bf if} $x$ is not adjacent to $H$ {\bf then} similar as step~4 
    ($H$ and $x$ make a $C^*_5$);
    \\
    6.2.\>\> {\bf if} $x$ has a single neighbor $h_i$ in $H$ {\bf then return} 
    claw    $\{h_i, x, h_{i-1}, h_{i+1}\}$; 
    \\
    6.3.\>\> {\bf if} $x$ is only adjacent to  $h_i, h_{i + 1}$ in $H$ {\bf then}
    \\
    \>\>\>{\bf if} $x v\in E(G)$ {\bf then return} claw $\{v, h_{i - 1}, h_{i + 2}, x\}$;
    \\
    \>\>\> {\bf else return} $S_3$ $\{x, h_i, h_{i + 1}, h_{i - 1}, v, h_{i + 2}\}$;
    \\
    6.4.\>\> {\bf if} $x$ is adjacent to $h_{i - 1}, h_{i + 1}$ but not $h_{i}$ {\bf 
      then return} $x h_{i - 1} h_i h_{i + 1}$ as a $C_4$;
    \\
    6.5.\>\> {\bf if} $x$ is adjacent to all vertices in $H$ {\bf then}
    \\
    \>\>\> {\bf if} $x y\not\in E(G)$ for some $y\in K_v$ or $K_i$
    {\bf then return} $x h_{i - 1} y h_{i + 1}$ as a $C_4$; 
    \\
    \>\>\> {\bf else} add $x$ to $K_v$;
    \>\comment{$x$ is adjacent to all vertices in the six cliques.}
    \\
    6.6.\>\> {\bf else} \>\> \comment{\;Hereafter 
      $|N(x)\cap H| = 3$; let them be $h_{i - 1}, h_i, h_{i + 1}$.}
    \\
    \>\>\> {\bf if} $x y\not\in E(G)$ for some $y\in K_i$ or $K_v$
    {\bf then return} $x h_{i - 1} y h_{i + 1}$ as a $C_4$; 
    \\
    \>\>\> {\bf if} $x y\not\in E(G)$ for some $y\in K_{i-1}$ 
    {\bf then return} claw $\{v, y, h_{i + 2}, x\}$;
    \\
    \>\>\> {\bf if} $x y\not\in E(G)$ for some $y\in K_{i+1}$
    {\bf then return} claw $\{v, h_{i - 2}, y, x\}$;
    \\
    \>\>\> {\bf if} $x y\in E(G)$ for some $y\in K_{i-2}$
    {\bf then return} $x h_{i + 1} h_{i + 2} y$ as a $C_4$; 
    \\
    \>\>\> {\bf if} $x y\in E(G)$ for some $y\in K_{i+2}$
    {\bf then return} $x h_{i - 1} h_{i - 2} y$ as a $C_4$; 
    \\
    \>\>\> {\bf else} add $x$ to $K_i$;
    \>\comment{$K_v, K_{i-1}, K_i, K_{i+1}\subseteq N(x)$ and $K_{i-2}, 
      K_{i+2}\cap N(x)=\emptyset$.}
    \\
    7.\> {\bf return} the six cliques.
  \end{tabbing}  

\end{quote} \vspace*{-6mm} \strut} $$\boxit{\box4}$$
\vspace*{-9mm}
\caption{Recognizing $\cal F$-free graphs.}
\label{fig:alg-recognition}
\end{figure}
\begin{proof}
  We prove only assertion (2) using the algorithm described in
  Fig.~\ref{fig:alg-recognition}, and its correctness implies
  assertion (1).  The algorithm starts by calling the certifying
  algorithm of Cao et al.~\cite{cao-15-nhcag} for recognizing proper
  Helly circular-arc graphs (step 0).  It enters one of steps~1--4,
  or~6 based on the outcome of step~0.  Here the subscripts of
  vertices in a hole $C_\ell$ should be understood to be modulo
  $\ell$.

  If the condition of any of steps~1--4 is satisfied, then either a
  proper and Helly arc model or a subgraph in $\cal F$ is returned.
  The correctness of steps~1--3 is straightforward.  Step~4.1 can find
  the path because $G$ is connected; possibly $v = x$, which is
  irrelevant in steps~4.2--4.4.  Note also that $\ell > 4$ in step~4.

  By Theorem~\ref{thm:phcag}, the algorithm passes steps~1--4 only
  when the outcome of step~0 is a $W_5$; let $H$ be its hole and let
  $v$ be the other vertex.  Steps 5--7 either detect an induced
  subgraph of $G$ in $\cal F$ or partition $V(G)$ into six cliques
  constituting a fat $W_5$.  Step~6 scans vertices not in the $W_5$
  one by one, and proceeds based on the adjacency between $x$ and $H$.
  In step~6.1, $H$ and $x$ make a $C^*_5$, which means that we can
  proceed exactly the same as step~4.  Note that the situation of
  step~6.4 is satisfied if $x$ is adjacent to four vertices of $H$.
  If none of the steps~6.1 to~6.5 applies, then $x$ has precisely
  three neighbors in $H$ and they have to be consecutive.  This is
  handled by step~6.6.

  Steps~0 and~4 take $O(m)$ time.  Steps~1, 2, 3, 5, and~7 need only
  $O(1)$ time.  {If the condition in step~6.1 is true, then it takes
    $O(m)$ time but it always terminates the algorithm after applying
    it.  Otherwise, step~6.1 is never called, and the rest of step~6
    scans the adjacency list of each vertex once, and hence takes
    $O(m)$ time in total.}  Therefore, the total running time of the
  algorithm is $O(m)$.  This concludes the proof.
\end{proof}

Implied by Theorem~\ref{thm:characterization}, a connected \{claw,
$S_3, \overline{S_3}$, $C_4$, $W_5$\}-free graph is a \phcag{}, which in turns
implies Proposition~\ref{lem:phcag}.  At this point a natural question
appealing to us is the relation between connected \{claw, $S_3, \overline{S_3}$,
$C_4$, $C_5$\}-free graphs and unit Helly circular-arc graphs.  Recall
that the class of unit interval graphs is a subclass of unit Helly
circular-arc graphs, on which we have a similar statement as
Corollary~\ref{lem:chordal-phcag}, i.e., a unit Helly circular-arc
graph that is chordal is a unit interval graph.  However, a connected
\{claw, $S_3, \overline{S_3}$, $C_4$, $C_5$\}-free graph that is not a unit Helly
circular-arc graph can be constructed as follows: Starting from a
$C_\ell$ with $\ell \ge 6$, for each edge $h_i h_{i+1}$ in the hole
add a new vertex $v_i$ and two new edges $v_i h_i, v_i h_{i+1}$.
(This is actually the CI($\ell,1$) graph defined by
Tucker~\cite{tucker-74-structures-cag}; see also
\cite{lin-13-nhcag-and-subclasses}.)  Therefore,
Proposition~\ref{lem:phcag} and Theorem~\ref{thm:characterization} are
the best we can expect in this sense.


Note that a $C_4$ is a \phcag{}.  Thus, the algorithm of
Theorem~\ref{thm:characterization} is not yet a certifying algorithm
for recognizing $\cal F$-free graphs.  To detect an induced $C_4$ from
a \phcag{}, we need to exploit its arc model.  If a \phcag{} $G$ is
not chordal, then the set of arcs for vertices in a hole necessarily
covers the circle, and it is minimal.  Interestingly, the converse
holds true as well---note that this is not true for chordal graphs.
\begin{proposition}\cite{lin-13-nhcag-and-subclasses, cao-15-nhcag}
  \label{lem:non-chordal-model}
  Let $G$ be a proper Helly circular-arc graph.  If $G$ is not
  chordal, then at least four arcs are needed to cover the whole
  circle in any arc model for $G$.
\end{proposition}
Proposition~\ref{lem:non-chordal-model} forbids among others two arcs
from having two-part intersection.\footnote{A model having no such
  intersection is called \emph{normal}; see the appendix for more
  discussion.}
\begin{corollary}\label{cor:non-chordal-model}
  Let $G$ be a proper Helly circular-arc graph that is not chordal and
  let $\cal A$ be an arc model for $G$.  A set of arcs inclusion-wise
  minimally covers the circle in $\cal A$ if and only if the vertices
  represented by them induce a hole of $G$.
\end{corollary}
Therefore, to find a shortest hole from a \phcag{}, we may work on an
arc model of it, and find a minimum set of arcs covering the circle in
the model.  This is another important step of our algorithm for the
\uie{} problem.  It has the detection of $C_4$'s as a special case,
because a $C_4$, if existent, must be the shortest hole of the graph.
\begin{lemma}\label{lem:shortest-hole}
  There is an $O(m)$-time algorithm for finding a shortest hole of a
  \phcag{}.
\end{lemma}
Before proving Lemma~\ref{lem:shortest-hole}, we need to introduce
some notation.  In an interval model, the interval $I(v)$ for vertex
$v$ is given by $[\lp{v}, \rp{v}]$, where \lp{v} and \rp{v} are its
\emph{left and right endpoints} respectively.  It always holds $\lp{v}
< \rp{v}$.  In an arc model, the arc $A(v)$ for vertex $v$ is given by
$[\ccp{v}, \cp{v}]$, where $\ccp{v}$ and $\cp{v}$ are its
\emph{counterclockwise and clockwise endpoints} respectively.  All
points in an arc model are assumed to be nonnegative; in particular,
they are between $0$ (inclusive) and $\ell$ (exclusive), where $\ell$
is the perimeter of the circle.  We point out that possibly
$\ccp{v}>\cp{v}$; such an arc $A(v)$ necessarily passes through the
point $0$.  Note that rotating all arcs in the model does not change
the intersections among them.  Thus we can always assume that a
particular arc contains or avoids the point $0$.  We say that an arc
model (for an $n$-vertex circular-arc graph) is \emph{canonical} if
the perimeter of the circle is $2 n$, and every endpoint is a
different integer in $\{0, 1, \ldots, 2 n - 1\}$.  Given an arc model,
we can make it canonical in linear time: We sort all $2 n$ endpoints
by radix sort, and replace them by their indices in the order.

Each point $\alpha$ in an interval model $\cal I$ or arc model $\cal
A$ defines a clique, denoted by $K_{\cal I}(\alpha)$ or $K_{\cal
  A}(\alpha)$ respectively, which is the set of vertices whose
intervals or arcs contain $\alpha$.  There are at most $2 n$ distinct
cliques defined as such.  If the model is Helly, then they include all
maximal cliques of this graph \cite{gavril-74-algorithms-cag}.  Since
the set of endpoints is finite, for any point $\rho$ in an interval or
arc model, we can find a small positive value $\epsilon$ such that
there is no endpoint in $[\rho-\epsilon, \rho)\cup (\rho,
\rho+\epsilon]$,---in other words, there is an endpoint in
$[\rho-\epsilon, \rho+\epsilon]$ if and only if $\rho$ itself is an
endpoint.  Note that the value of $\epsilon$ should be understood as a
function, depending on the interval/arc model as well as the point
$\rho$, instead of a constant.

Let $G$ be a non-chordal graph and let $\cal A$ be a proper and Helly
arc model for $G$.  If $u v\in E(G)$, then exactly one of $\ccp{v}$
and $\cp{v}$ is contained in $A(u)$
(Proposition~\ref{lem:non-chordal-model}).  Thus, we can define a
``left-right relation'' for each pair of intersecting arcs, which can
be understood from the viewpoint of an observer placed at the center
of the model.  We say that arc $A(v)$ intersects arc $A(u)$ from the
left when $\cp{v}\in A(u)$, denoted by $v\rightarrow u$.  Any set of
arcs whose union is an arc not covering the circle (the corresponding
vertices induce a connected unit interval graph) can be ordered in a
unique way such that $v_i \rightarrow v_{i+1}$ for all $i$.  From it
we can find the leftmost (most counterclockwise) and rightmost (most
clockwise) arcs.

For any vertex $v$, let $h(v)$ denote the length of the shortest holes
through $v$, which is defined to be $+\infty$ if no hole of $G$
contains $v$.  The following is important for the proof of
Lemma~\ref{lem:shortest-hole}.

\begin{lemma}\label{lem:greedy-hole}
  Let $\cal A$ be a proper and Helly arc model for a non-chordal graph
  $G$.  Let $v_1, v_2, \ldots, v_p$ be a sequence of vertices such
  that for each $i = 2, \ldots,p$, the arc $A(v_i)$ is the rightmost
  of all arcs containing \cp{v_{i - 1}}.  If $v_1$ is contained in
  some hole and $v_i v_1\not\in E(G)$ for all $2 < i \le p$, then
  there is a hole of length $h(v_1)$ containing $v_1, v_2, \ldots,
  v_p$ as consecutive vertices on it.
\end{lemma}
\begin{proof} 
  Suppose that there is no such a hole, then there exists a smallest
  number $i$ with $2\le i\le p$ such that no hole of length $h(v_1)$
  contains $v_1, v_2, \ldots, v_i$.  By assumption, there is a hole
  $v_1 \cdots v_{i - 1} u_i u_{i+1}\cdots u_{h(v_1)}$ of length
  $h(v_1)$ with $u_i \ne v_i$; let it be $H$.  In case that $i = 2$,
  we assume that $H$ is given in the way that $v_1 \rightarrow u_2$.
  By Corollary~\ref{cor:non-chordal-model}, the set of arcs for $H$
  cover the circle in $\cal A$.  Since $v_{i - 1} \rightarrow u_i$ and
  by the assumption on $v_i$, the arc $A(v_i)$ covers $[\cp{v_{i -1}},
  \ccp{u_{i+1}}]$ (note that $i < h(v_1)$ as otherwise $v_i\rightarrow
  v_1$).  Therefore, the arcs for $V(H)\setminus \{u_i\}\cup \{v_i\}$
  cover the circle as well.  By Corollary~\ref{cor:non-chordal-model},
  a subset of these vertices induces a hole; since $v_i$ is not
  adjacent to $v_1$ from the left, this subset has to contain $v_1$.
  But a hole containing $v_1$ cannot be shorter than $H$, and hence it
  contains $v_1, \ldots, v_i$, a contradiction.  Therefore, such an
  $i$ does not exist, and there is a hole of length $h(v_1)$
  containing $v_1, v_2, \ldots, v_p$.

  Since $v_j$ and $v_{j+1}$ are adjacent for all $1\le j\le p-1$,
  vertices $v_1, v_2, \ldots, v_p$ have to be consecutive on this
  hole.
\end{proof}

\begin{figure}[h!]
\setbox4=\vbox{\hsize28pc \noindent\strut
\begin{quote}
  \vspace*{-5mm}

  {\bf Algorithm} {\bf shortest-hole}($G, \cal A$)
  \\
  {\bf Input}: a proper and Helly arc model $\cal A$ for a non-chordal
  graph $G$.
  \\
  {\bf Output}: a shortest hole of $G$.
  \begin{tabbing}
    AAA\=aaA\=Aaa\=MMMMMMAAAAAAAAAAAAAAAAAAAAAAAAA\=A \kill
    0.\> make $\cal A$ canonical where $0$ is $\ccp{v}$ for some $v$; 
    \\
    1.\> {\bf for} $i = 1,\ldots, \cp{v} - 1$ {\bf do}
    \\
    1.1.\>\>{\bf if} $i$ is $\ccp{x}$ {\bf then} create a new array $\{x\}$;
    \\
    \textbackslash\!\textbackslash\; {\it these arrays are 
      circularly linked so that the next of the last array is the 
      first one.}
    \\
    2.\> $w \leftarrow \bot$; $U \leftarrow$ the first array; 
    \\
    3.\>{\bf for} $i = \cp{v} + 1,\ldots, 2n - 1$ {\bf do}
    \\
    3.0.\>\> $z \leftarrow$ the last vertex of $U$;
    \\
    3.1.\>\> {\bf if} $i$ is $\ccp{x}$ {\bf then} $w \leftarrow
    x$; {\bf continue}$^\dag$;
    \\
    3.2.\>\>{\bf if} $i \ne \cp{z}$ {\bf then continue};
    \\
    3.3.\>\>{\bf if} $w = \bot$ {\bf then} delete $U$; $U \leftarrow$ 
    the next array of $U$; 
    \\
    3.4.\>\>{\bf if} $w \ne \bot$ {\bf then} append $w$ to $U$; 
    $w \leftarrow \bot$; $U \leftarrow$ the next array of $U$; 
    \\
    4.\> {\bf for each} $U$ till the last array {\bf do}
    \\
    4.1.\>\> {\bf if} the first and last vertices of $U$ are adjacent {\bf then 
      return} $U$;
    \\
    5.\> {\bf return} $U\cup \{v\}$.  \>\>\>\comment{$U$ is the last array.}
    \\[2mm]
    \dag: i.e., ignore the rest of this iteration of the for-loop 
    and proceed to the next iteration.
  \end{tabbing}  

\end{quote} \vspace*{-6mm} \strut} $$\boxit{\box4}$$
\vspace*{-9mm}
\caption{Finding a shortest hole in a \phcag{}.}
\label{fig:alg-shortest-hole}
\end{figure}
Let $\alpha$ be any fixed point in a proper and Helly arc model.
According to Corollary~\ref{cor:non-chordal-model}, every hole needs
to visit some vertex in $K_{\cal A}(\alpha)$.  Therefore, to find a
shortest hole in $G$, it suffices to find a hole of length $\min\{h(x)
: x\in K_{\cal A}(\alpha)\}$.
\begin{proof}[Proof of Lemma~\ref{lem:shortest-hole}]
  The algorithm described in Fig.~\ref{fig:alg-shortest-hole} finds a
  hole of length $\min\{h(x) : x\in K_{\cal A}(\cp{v} + 0.5)\}$.  Step
  1 creates $|K_{\cal A}(\cp{v} + 0.5)|$ arrays, each starting with a
  distinct vertex in $K_{\cal A}(\cp{v} + 0.5)$, and they are ordered
  such that their (counter)clockwise endpoints are increasing.  The
  main job of this algorithm is done in step~3.  During this step, $w$
  is the new vertex to be processed, and $U$ is the current array.
  Each new vertex is added to at most one array, while each array is
  either dropped or extended.  We use $\bot$ as a dummy vertex, which
  means that no new vertex has been met after the last one has been
  put into the previous array.  Step~3.1 records the last scanned arc.
  Once the clockwise endpoint of the last vertex $z$ of the current
  array $U$ is met, $w$ is appended to $U$ (step~3.4); note that
  $A(w)$ is the most clockwise arc that contains \cp{z}.  On the other
  hand, if $w = \bot$, then we drop this array from further
  consideration (step~3.3).  If after step~3, one of the arrays
  already induces a hole (i.e., the first and last vertices are
  adjacent), then it is returned in step~4.1.  Otherwise, $U$ does not
  induces a hole, and step~5 returns the hole induced by $U$ and $v$.

  We now verify the correctness of the algorithm.  It suffices to show
  that the length of the found hole is $\min\{h(x) : x\in K_{\cal
    A}(\cp{v} + 0.5)\}$.  The following hold for each array $U$:
  \begin{enumerate}[(1)]
  \item For any pair of consecutive vertices $u, w$ of $U$, the arc
    $A(w)$ is the rightmost of all arcs containing \cp{u}.
  \item At the end of the algorithm, if $U$ is not dropped, then $0 <
    \cp{z} < \cp{v}$, where $z$ is the last vertex of $U$.
  \end{enumerate}
  Let $v_1, v_2, \ldots, v_p$ be the vertices in $U$.  From (1) and
  (2) it can be inferred that for $1< i< p$, the vertex $v_i$ is
  adjacent to only $v_{i-1}, v_{i+1}$ in $U$, while $v_1$ and $v_p$
  may or may not be adjacent.  If $v_1 v_p\in E(G)$, then vertices in
  $U$ induce a hole of $G$; otherwise $U \cup\{v\}$ induces a hole of
  $G$.  By Lemma~\ref{lem:greedy-hole}, this hole has length $h(v_1)$.

  Some of the $|K_{\cal A}(\cp{v} + 0.5)|$ arrays created in step~1
  may be dropped in step~3.  Note that step~3 processes arrays in a
  circular order starting from the first one, and each array is either
  deleted (step~3.3) or extended by adding one vertex (step~3.4).
  \begin{enumerate}[(1)]
    \setcounter{enumi}{2}
  \item At any moment of the algorithm, the sizes of any two arrays
    differ by at most one.  In particular, at the end of step~3, if
    the current array $U$ is not the first, then $U$, as well as all
    the succeeding arrays, has one less element than its
    predecessor(s).
  \end{enumerate}
  This ensures that the hole returned in step~4.1 or 5 is the shortest
  among all the holes decided by the remaining arrays after step~3.

  It remains to argue that for any array $U$ deleted in step~3.3, the
  found hole is not longer than $h(x)$, where $x$ is the first vertex
  of $U$.  All status in the following is referred at the moment
  before $U$ is deleted (i.e., before step~3.3 that deletes $U$).  Let
  $z$ be the last vertex of $U$.  Let $U'$ be the array that is
  immediately preceding $U$, and let $z'$ be its last vertex.  Note
  that there is no arc with a counterclockwise endpoint between
  $\cp{z'}$ and $\cp{z}$, as otherwise $w\ne \bot$ and $U$ would not
  be deleted.  Therefore, any arc intersecting $A(z)$ from the right
  also intersects $A({z'})$ from the right.  By
  Lemma~\ref{lem:greedy-hole}, there is a hole $H$ that has length
  $h(x)$ and contains $U\cup \{z'\}$.  We find a hole through $U'$ of
  the same length as follows.  If $U$ is not the first array, then
  $|U'| = |U| + 1$, and we replace $U\cup \{z'\}$ by $U'$.  Otherwise,
  $|U'| = |U|$, and we replace $U\cup \{z'\}$ by $\{v\}\cup U'$.  It
  is easy to verify that after this replacement, $H$ remains a hole of
  the same length.

  We now analyze the running time of the algorithm.  Each of the $2 n$
  endpoints is scanned once, and each vertex belongs to at most one
  array.  Using a linked list to store an array, the addition of a new
  vertex can be implemented in constant time.  Using a circularly
  linked list to organize the arrays, we can find the next array or
  delete the current one in constant time.  With the (endpoints of)
  all arcs given, the adjacency between any pair of vertices can be
  checked in constant time, and thus step~4 takes $O(n)$ time.  It
  follows that the algorithm can be implemented in $O(m)$ time.
\end{proof}

\section{Vertex deletion}\label{sec:vertex-deletion}

We say that a set $V_-$ of vertices is a \emph{hole cover} of $G$ if
$G - V_-$ is chordal.  The hole covers of proper Helly circular-arc
graphs are characterized by the following lemma.
\begin{lemma}\label{lem:hole-cover}
  Let $\cal A$ be a proper and Helly arc model for a non-chordal graph
  $G$.  A set $V_- \subseteq V(G)$ is a hole cover of $G$ if and only if
  it contains $K_{\cal A}(\alpha)$ for some point $\alpha$ in $\cal A$.
\end{lemma}
\begin{proof}
  For any vertex set $V_-$, the subgraph $G - V_-$ is also a proper
  Helly circular-arc graph, and the set of arcs $\{A(v) : v\in
  V(G)\setminus V_-\}$ is a proper and Helly arc model for $G - V_-$.
  For the ``if'' direction, we may rotate $\cal A$ to make $\alpha =
  0$, and then setting $I(v) = A(v)$ for each $v\in V(G)\setminus V_-$
  gives a proper interval model for $G - V_-$.  For the ``only if''
  direction, note that if there is no $\alpha$ with $K_{\cal
    A}(\alpha)\subseteq V_-$, then we can find a minimal set $X$ of
  vertices from $V(G)\setminus V_-$ such that $\bigcup_{v\in X} A(v)$
  covers the whole circle in $\cal A$.  According to
  Corollary~\ref{cor:non-chordal-model}, $X$ induces a hole of $G$,
  which remains in $G - V_-$.
\end{proof}

Noting that any ``local'' part of a proper and Helly arc model
``behaves similarly'' as an interval model, Lemma~\ref{lem:hole-cover}
is an easy extension of the clique separator property of interval
graphs \cite{fulkerson-65-interval-graphs}.  On the other hand, to get
a unit interval graph out of a fat $W_5$, it suffices to delete a
smallest clique from the fat hole.  Therefore,
Theorem~\ref{thm:characterization} and Lemma~\ref{lem:hole-cover}
imply the following linear-time algorithm.
\begin{corollary}\label{lem:vertex-deletion-phcag}
  The \uivd{} problem can be solved in $O(m)$ time (1) on \phcag{s}
  and (2) on $\cal F$-free graphs.
\end{corollary}
We are now ready to prove the main results of this section.
\begin{theorem}\label{thm:algs-uivd}
  There are an $O(6^k\cdot m)$-time parameterized algorithm for the
  \uivd{} problem and an $O(n m)$-time approximation algorithm of
  approximation ratio $6$ for its minimization version.
\end{theorem}
\begin{proof}
  Let ($G, k$) be an instance of \uivd{}; we may assume that $G$ is
  not a unit interval graph and $k > 0$.  The parameterized algorithm
  calls first Theorem~\ref{thm:characterization}(2) to decide whether
  it has an induced subgraph in $\cal F$, and then based on the
  outcome, it solves the problem by making recursive calls to itself,
  or calling the algorithm of
  Corollary~\ref{lem:vertex-deletion-phcag}.  If an induced subgraph
  $F$ in $\cal F$ is found, it calls itself $|V(F)|$ times, each with
  a new instance ($G - v, k - 1$) for some $v\in V(F)$; since we need
  to delete at least one vertex from $V(F)$, the original instance
  ($G, k$) is a yes-instance if and only if at least one of the
  instances ($G - v, k - 1$) is a yes-instance.  Otherwise, $G$ is
  $\cal F$-free and the algorithm calls
  Corollary~\ref{lem:vertex-deletion-phcag} to solve it.  The
  correctness of the algorithm follows from discussion above and
  Corollary~\ref{lem:vertex-deletion-phcag}.  On each subgraph in
  $\cal F$, which has at most $6$ vertices, at most $6$ recursive
  calls are made, all with parameter value $k - 1$.  By
  Theorem~\ref{thm:characterization}, each recursive call is made in
  $O(m)$ time; each call of Corollary~\ref{lem:vertex-deletion-phcag}
  takes $O(m)$ time.  Therefore, the total running time is $O(6^k\cdot
  m)$.

  The approximation algorithm is adapted from the parameterized
  algorithm as follows.  For the subgraph $F$ found by
  Theorem~\ref{thm:characterization}, we delete all its vertices.  We
  continue the process until the remaining graph is $\cal F$-free, and
  then we call Corollary~\ref{lem:vertex-deletion-phcag} to solve it
  optimally.  Each subgraph in $\cal F$ has $4$ or $6$ vertices, and thus
  at most $n/4$ such subgraphs can be detected and deleted, each
  taking $O(m)$ time, hence $O(n m)$ in total.
  Corollary~\ref{lem:vertex-deletion-phcag} takes another $O(m)$ time.
  The total running time is thus $O(n m) + O(m) = O(nm)$, and the
  approximation ratio is clearly $6$.
\end{proof}

\section{Edge deletion}\label{sec:edge-deletion}
Inspired by Lemma~\ref{lem:hole-cover}, one may expect a similarly
nice characterization---being ``local'' to some point in an arc model
for $G$---for a minimal set of edges whose deletion from a
\phcag{}~$G$ makes it chordal.  This is nevertheless not the case; as
shown in Fig.~\ref{fig:lem:edge-hole-cover}, they may behave in a very
strange or pathological way.
\begin{SCfigure}[][h]
  \centering
    \includegraphics{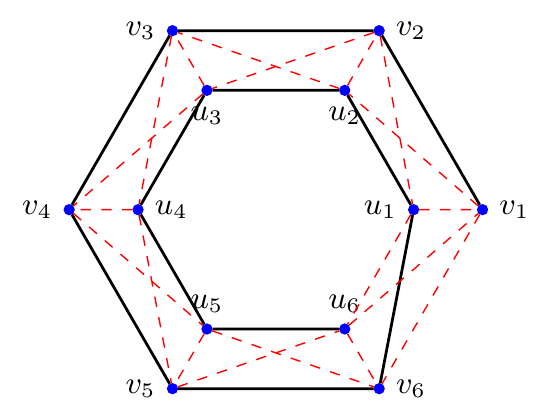} 
  \caption{\small The set of all 30 edges (both solid and
    dashed) spans a \phcag{}.  After the set of 19 dashed edges
    deleted (we rely on the reader to verify its minimality), the
    remaining graph on the 11 solid edges is a unit interval graph.
    \\Note that four edges would suffice, e.g., $\{u_2 u_3, u_2 v_3,
    v_2 u_3, v_2 v_3\}$.}
  \label{fig:lem:edge-hole-cover}
\end{SCfigure}

Recall that $v\rightarrow u$ means arc $A(v)$ intersecting arc $A(u)$
from the left, or $\cp{v}\in A(u)$.  For each point~$\alpha$ in a
proper and Helly arc model $\cal A$, we can define the following set
of edges:
\begin{equation}
    \label{eq:2}
    \overrightarrow E_{\cal A}(\alpha) = \{v u: v\in K_{\cal A}(\alpha),\, u\not\in K_{\cal A}(\alpha), v\rightarrow u\}.
\end{equation}
One may symmetrically view $\overrightarrow E_{\cal A}(\alpha)$ as
$\{v u: v\not\in K_{\cal A}(\beta),\, u\in K_{\cal A}(\beta),
v\rightarrow u\}$, where $\beta := \max\{\cp{x}: x\in K_{\cal
  A}(\alpha) \} + \epsilon$.  It is easy to verify that the following
gives a proper interval model for $G - \overrightarrow E_{\cal A}(0)$:
\begin{equation}
  \label{eq:cag-to-im}
  I(v) := 
  \begin{cases}
    [\ccp{v}, \cp{v} + \ell] & \text{if } v\in K_{\cal A}(0),
    \\
    [\ccp{v}, \cp{v}] & \text{otherwise},
  \end{cases}
\end{equation}
where $\ell$ is the perimeter of the circle in $\cal A$; see
Fig.~\ref{fig:arc-model-to-interval}.  For an arbitrary point
$\alpha$, the model $G - \overrightarrow E_{\cal A}(\alpha)$ can be
given analogously, e.g., we may rotate the model first to make $\alpha
= 0$.
\begin{figure*}[h]
  \centering\small
  \subfloat[An arc model for $G$.]{\label{fig:cut-points}
    \includegraphics{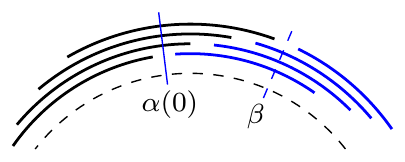} 
  }
  \qquad
  \subfloat[The interval model for $G - \protect\overrightarrow E_{\cal A}(\alpha)$ given by \eqref{eq:cag-to-im}.]{\label{fig:after-cut}
    \includegraphics{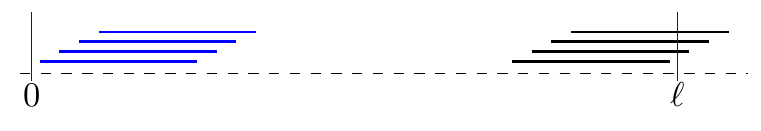} 
  }
  \caption{Illustration for Proposition~\ref{lem:edge-hole-cover-1}.}
  \label{fig:arc-model-to-interval}
\end{figure*}

\begin{proposition}\label{lem:edge-hole-cover-1}
  Let $\cal A$ be a proper and Helly arc model for a non-chordal graph
  $G$.  For any point $\alpha$ in $\cal A$, the subgraph $G -
  \overrightarrow E_{\cal A}(\alpha)$ is a unit interval graph.
\end{proposition}

The other direction is more involved and very challenging.  A unit
interval graph $\underline G$ is called a \emph{spanning unit interval
  subgraph} of $G$ if $V(\underline G) = V(G)$ and $E(\underline
G)\subseteq E(G)$; it is called {\em maximum} if it has the largest
number of edges among all spanning unit interval subgraphs of $G$.  To
prove that all maximum spanning unit interval subgraphs have a certain
property, we use the following argument by contradiction.  Given a
spanning unit interval subgraph $\underline G$ not having the
property, we locally modify a unit interval model $\cal I$ for
$\underline G$ to a {\em proper} interval model $\cal I'$ such that
the represented graph $\underline G'$ satisfies $E(\underline
G')\subseteq E(G)$ and $|E(\underline G')| > |E(\underline G)|$.
Recall that we always select $\epsilon$ in a way that there cannot be
any endpoint in $[\rho - \epsilon, \rho)\cup (\rho, \rho + \epsilon]$,
and thus an arc covering $\rho + \epsilon$ or $\rho - \epsilon$ must
contain $\rho$.

\begin{lemma}\label{lem:edge-hole-cover}
  Let $\cal A$ be a proper and Helly arc model for a non-chordal graph
  $G$.  For any maximum spanning unit interval subgraph $\underline G$
  of $G$, the deleted edges, $E(G)\setminus E(\underline G)$, are
  $\overrightarrow E_{\cal A}(\rho)$ for some point $\rho$ in $\cal
  A$.
\end{lemma}
\begin{figure*}[t]
  \centering\small
  \subfloat{\label{fig:a}
    \includegraphics{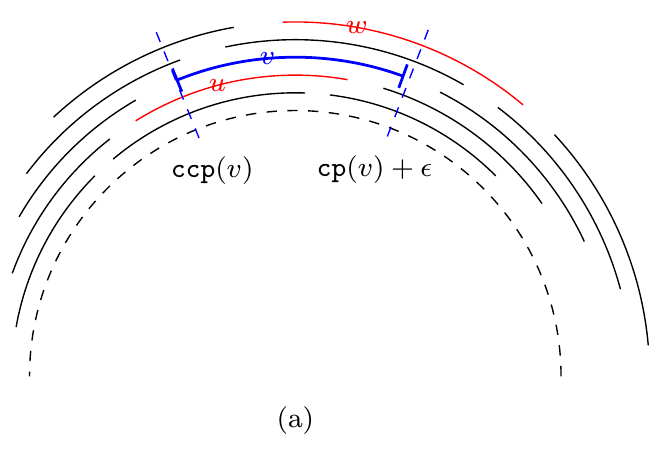} 
  }

  \subfloat[]{\label{fig:b}
    \includegraphics{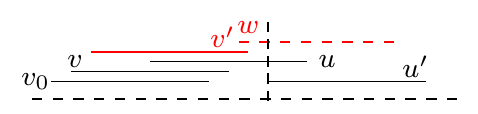} 
  }
  \qquad
  \subfloat[]{\label{fig:c}
    \includegraphics{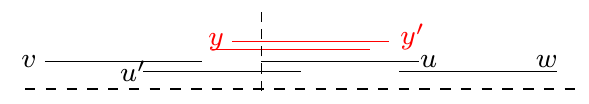} 
  }

  \subfloat[]{\label{fig:d}
    \includegraphics{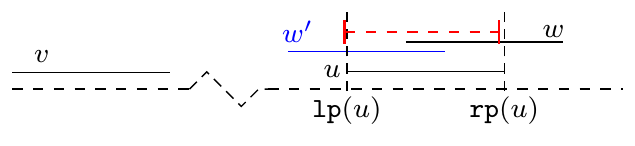} 
  }
  \qquad
  \subfloat[]{\label{fig:e}
    \includegraphics{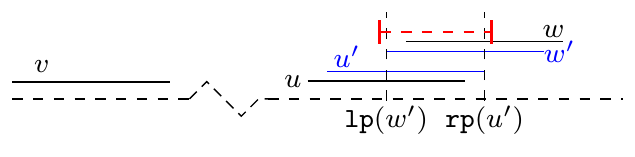} 
  }
  \caption{Illustration for the proof of Lemma~\ref{lem:edge-hole-cover}.}
  \label{fig:minimum-edge-hole-cover}
\end{figure*}

\begin{proof}
  Let $\cal I$ be a unit interval model for $\underline G$, and let
  $E_- = E(G)\setminus E(\underline G)$, i.e., the set of deleted
  edges from $G$.  We find first a vertex $v$ satisfying at least one
  of the following conditions.

  \begin{enumerate}[(C1)]
  \item The sets $N_{\underline G}[v]$ and $K_{\cal A}(\ccp{v} -
    \epsilon)$ are disjoint and all edges between them are in $E_-$.
  \item The sets $N_{\underline G}[v]$ and $K_{\cal A}(\cp{v} +
    \epsilon)$ are disjoint and all edges between them are in $E_-$.
  \end{enumerate}
  Recall that a vertex $u\in K_{\cal A}(\ccp{v} - \epsilon)$ if and
  only if $u \rightarrow v$, and a vertex $w\in K_{\cal A}(\cp{v} +
  \epsilon)$ if and only if $v \rightarrow w$; see Fig.~\ref{fig:a}.
  These two conditions imply $N_{\underline G}[v]\subseteq K_{\cal
    A}(\cp{v})$ and $N_{\underline G}[v]\subseteq K_{\cal A}(\ccp{v})$
  respectively: All edges between $v$ itself, which belongs to
  $N_{\underline G}[v]$, and $K_{\cal A}(\ccp{v} - \epsilon)$ or
  $K_{\cal A}(\cp{v} + \epsilon)$ are in $E_-$.

  Let $I(v_0)$ be the leftmost interval in $\cal I$.  Note that by
  Proposition~\ref{lem:non-chordal-model}, arcs for $N_{\underline
    G}[v_0]$ cannot cover the whole circle.  If $N_{\underline
    G}[v_0]$ is a separate component in $\underline G$, then we can
  take $v$ to be the vertex from $N_{\underline G}[v_0]$ with the
  leftmost arc; it satisfies condition (C1).  Otherwise, let $I(u)$ be
  the last interval containing $\rp{v_0}$ and let $I(u')$ be the next
  interval (i.e., $I(u')$ is the leftmost interval that does not
  intersect $I(v_0)$).  See Fig.~\ref{fig:b}.  Intervals $I(u)$ and
  $I(u')$ intersect ($N_{\underline G}[v_0]$ is not isolated), but
  $v_0 u'\not\in E(G)$ because $\underline G$ is maximum: Moving
  $I(v_0)$ to the right to intersect $I(u')$ would otherwise make a
  unit interval model that represents a subgraph of $G$ having one
  more edge than $\underline G$.
  \begin{itemize}
  \item If $u'\rightarrow u$, then $u'\rightarrow u\rightarrow v_0$.
    We argue by contradiction that there cannot be vertices $v'\in
    K_{\cal A}(\cp{v_0} + \epsilon)$ and $w\not\in N_G(v_0)$ with $u
    v', v' w\in E(\underline G)$.  If $u v'\in E(\underline
    G)\subseteq E(G)$, then by the position of $A(v_0)$, we must have
    $u\rightarrow v'$ and $v_0\rightarrow v'$.  Then $v_0 u' \not\in
    E(G)$ excludes the possibility $u'\rightarrow v'$; on the other
    hand, $v'\rightarrow u'$ is excluded by
    Proposition~\ref{lem:non-chordal-model}: The arcs for $u', u, v'$
    cannot cover the whole circle.  Therefore, $v' u'\not\in E(G)$,
    and likewise $u w\not\in E(G)$.  They cannot be in $E(\underline
    G)$ either, but this is impossible; See Fig.~\ref{fig:b}.  Let $v$
    be the vertex in $N_{\underline G}[v_0]$ such that $A(v)$ is the
    rightmost; it satisfies condition (C2).
  \item If $u\rightarrow u'$, then $v_0\rightarrow u\rightarrow u'$.
    Similarly as above, the vertex $v$ in $N_{\underline G}[v_0]$ such
    that $A(v)$ is the leftmost satisfies condition (C1).
  \end{itemize}

  Noting that conditions (C1) and (C2) are symmetric, we assume that the
  vertex $v$ found above satisfies condition (C2), and a symmetric
  argument would apply to condition (C1).
  Note that by the selection of $v$, which has the leftmost arc among
  $N_{\underline G}[v_0]$, we have $N_{\underline G}[v] \subseteq
  N_{G}[v']$ for every $v'$ satisfying $\lp{v'} < \lp{v}$; thus
  setting their intervals to be $[\lp{v} + \epsilon, \rp{v} +
  \epsilon]$ would produce another unit interval model for $\underline
  G$.  In the rest of the proof we may assume without loss of
  generality that $I(v)$ is the first interval in $\cal I$.

  Since the model $\cal A$ is proper and Helly, no arc in $\cal A$ can
  contain both $\ccp{v}$ and $\cp{v} + \epsilon$.  In other words,
  $K_{\cal A}(\ccp{v})$ and $K_{\cal A}(\cp{v} + \epsilon)$ are
  disjoint, and $\overrightarrow E_{\cal A}(\ccp{v})$ comprises
  precisely edges between them; see Fig.~\ref{fig:a}.
  By Proposition~\ref{lem:edge-hole-cover-1}, $|E_-| \le
  |\overrightarrow E_{\cal A}(\ccp{v})|$.  If $K_{\cal A}(\ccp{v})$
  is not adjacent to $K_{\cal A}(\cp{v} + \epsilon)$ in $\underline
  G$, then $\overrightarrow E_{\cal A}(\ccp{v})\subseteq E_-$, and
  they have to be equal.  In this case the proof is complete:
  $\rho = \ccp{v}$.  
  We are hence focused on edges between $K_{\cal A}(\ccp{v})$ and
  $K_{\cal A}(\cp{v} + \epsilon)$.

  \begin{claim}
    Let $u\in K_{\cal A}(\ccp{v})$.  If $u$ is adjacent to $K_{\cal
      A}(\cp{v} + \epsilon)$ in $\underline G$, then $I(u)$ does not
    intersect the intervals for $N_{\underline G}[v]$.
  \end{claim}
  \begin{proof}
    Recall that $uv\notin E(\underline G)$ by condition (C2) and the
    fact that $u$ is adjacent to $K_{\cal A}(\cp{v} + \epsilon)$ in
    $\underline G$.
    Let $u'$ be the vertex in $N_{\underline G}[v]$ with the rightmost
    interval, and let $I(y)$ be the leftmost interval not intersecting
    $I(v)$.  Note that $v y\not\in E(G)$: Otherwise moving $I(v)$ to
    the right to intersect $I(y)$ would make a unit interval model
    that represents a subgraph of $G$ with one more edge than
    $\underline G$.  Suppose to the contrary of this claim that $I(u)$
    intersects some interval for $N_{\underline G}[v]$, then it
    intersects $I(u')$.  See Fig.~\ref{fig:c}.  Since $v$ satisfies
    condition (C2) and by Proposition~\ref{lem:non-chordal-model},
    $u'\rightarrow v$ and $y \rightarrow u, u'$.

    Let $w$ be the vertex in $K_{\cal A}(\cp{v} + \epsilon)$ that has
    the leftmost interval.  Since $v$ satisfies condition (C2),
    $\lp{w} > \rp{u'}$; by the assumption that $u$ is adjacent to
    $K_{\cal A}(\cp{v} + \epsilon)$ in $\underline G$, we have $\lp{w}
    < \rp{u}$.  By Proposition~\ref{lem:non-chordal-model}, $y w
    \not\in E(G)$.
    Let $A(y')$ be the leftmost arc such that $y'\in K_{\cal
      A}(\ccp{u})$ and $I(y')$ lies in ($\rp{v}, \lp{w}$); this vertex
    exists because $y$ itself is a candidate for it.
    Again by Proposition~\ref{lem:non-chordal-model}, $y' w \not\in
    E(G)$.

    Let $X$ denote the set of vertices $x$ with $\lp{x} < \lp{w}$.  We
    make a new interval model by resetting intervals for these
    vertices.  Since every vertex in $X$ is adjacent to at least one
    of $u$ and $u'$, by Proposition~\ref{lem:non-chordal-model}, the
    union of arcs for $X$ do not cover the circle in $\cal A$.  These
    arcs can thus be viewed as a proper interval model for $G[X]$.
    The new intervals for $X$, adapted from these arcs, are formally
    specified as follows.  The left endpoint of each $x\in X$ is set
    as $\lpp{x} = \lp{w} - (\ccp{w} - \ccp{x})$.
    For each vertex $x\in X\setminus N_G(w)$, we set 
    \[
    I'(x) = [\lp{w} - (\ccp{w} - \ccp{x}), \lp{w} - (\ccp{w} -
    \cp{x})].
    \]
    By the selection of $w$, we have $x\rightarrow w$ for all $x\in
    X\cap N_G(w)$.  Arcs for $X\cap N_G(w)$ are thus pairwise
    intersecting; by Proposition~\ref{lem:non-chordal-model}, they
    cannot cover the whole circle.
    We can thus number vertices in $X\cap N_G(w)$ as $u_1,
    \ldots, u_p$ such that $u_i \rightarrow u_{i + 1}$ for each $i =
    1, \ldots, p - 1$.  Their right endpoints are set as
    \begin{align*}
      \rpp{u_1} &= \max\big\{ \lp{w} + \epsilon, \rp{u_1}\big\},
      \qquad\text{and}
      \\
      \rpp{u_i} &= \max\big\{ \rpp{u_{i - 1}} + \epsilon,
      \rp{u_i}\big\} \quad \text{ for } i = 2, \ldots, p.
    \end{align*}

    Let $\cal I'$ denote the resulting new interval model.  To see
    that $\cal I'$ is proper, note that (a) no new interval can
    contain or be contained by an interval $I(z)$ for $z\in
    V(G)\setminus X$; and (b) the left and right endpoints of the
    intervals for $X$ have the same ordering as the counterclockwise
    and clockwise endpoints of the arcs $\{A(x): x\in X\}$, hence
    necessarily proper.  
    Let $G'$ denote the proper interval graph represented by $\cal
    I'$.  We want to argue that $E(\underline G) \subset E(G')
    \subseteq E(G)$, which would contradict that $\underline G$ is a
    maximum unit interval subgraph of $G$, and conclude the proof of
    this claim.

    By construction, $G' - X$ is the same as $\underline G - X$, while
    $G'[X]$ is the same as $G[X]$.  Thus, we focus on edges between
    $X$ and $V(G)\setminus X$, which are all incident to $X\cap
    N_G(w)$.
    For each $i = 1, \ldots, p$, we have $\rpp{u_i} \ge \rp{u_i}$, and
    thus $E(\underline G)\subseteq E(G')$; on the other hand, they are
    not equal because $uv \in E(G')\setminus E(\underline G)$.  We
    show by induction that for every $i = 1, \ldots, p$, the edges
    incident to $u_i$ in $G'$ is a subset of $G$.  The base case is
    clear: $N_{G'}(u_1)\setminus X$ is either $\{w\}$ or
    $N_{\underline G}(u_1)\setminus X$.  For the inductive step, if
    $\rpp{u_i} = \rp{u_i}$, then $N_{G'}(u_i)\setminus X =
    N_{\underline G}(u_i)\setminus X\subseteq N_G(u_i)$; otherwise,
    $N_{G'}(u_i)\setminus X\subseteq N_{\underline G}(u_{i -
      1})\setminus X$, which is a subset of $N_G(u_i)$ because $u_{i -
      1}\rightarrow u_i$.  This verifies $E(G') \subseteq E(G)$.
    \renewcommand{\qedsymbol}{$\lrcorner$}
  \end{proof}

  We consider then edges deleted from each vertex in $K_{\cal
    A}(\ccp{v})$, i.e., edges in $E_-$ that is incident to $K_{\cal
    A}(\ccp{v})$.
  \begin{claim}
    Let $u\in K_{\cal A}(\ccp{v})$.  If $u$ is adjacent to $K_{\cal
      A}(\cp{v} + \epsilon)$ in $\underline G$, then there are
    strictly more edges incident to $u$ in $E_-$ than in
    $\overrightarrow E_{\cal A}(\ccp{v})$ (i.e., between $u$ and
    $K_{\cal A}(\cp{v} + \epsilon)$.).
  \end{claim}
  \begin{proof}
    The vertices in $N_G(u)$ consists of three parts, $N_{\underline
      G}[v]$, those in $K_{\cal A}(\cp{v} + \epsilon)$, and others.
    By Claim~1, edges between $u$ and all vertices in $N_{\underline
      G}[v]$ are in $E_-$; they are not in $\overrightarrow E_{\cal
      A}(\ccp{v})$.  All edges between $u$ and $K_{\cal A}(\cp{v} +
    \epsilon)$ are in $\overrightarrow E_{\cal A}(\ccp{v})$.  In
    $\overrightarrow E_{\cal A}(\ccp{v})$ there is no edge between $u$
    and the other vertices.  Therefore, to show the claim, it suffices
    to show $|N_{\underline G}[u]\cap K_{\cal A}(\cp{v} + \epsilon)|
    \le |N_{\underline G}(v)| < |N_{\underline G}[v]| $.

    Let $u w$ be an edge of $\underline G$ with $w\in K_{\cal
      A}(\cp{v} + \epsilon)$.  Since $v$ satisfies condition (C2),
    both $uv$ and $v w$ are in $E_-$.  We show only the case that
    $I(w)$ intersects $I(u)$ from the right; a similar argument works
    for the other case: Note that all the intervals used below are
    disjoint from $I(v)$.

    Consider first that there exists another vertex $w' \in K_{\cal
      A}(\cp{v} + \epsilon)$ with interval $I(w')$ intersecting $I(u)$
    from the left.  See Fig.~\ref{fig:d}.  Any interval that
    intersects $I(u)$ necessarily intersects at least one of $I(w)$
    and $I(w')$, and thus by Proposition~\ref{lem:non-chordal-model},
    its vertex must be in $N_G(v)$.  As a result, setting $I(v)$ to
    $[\lp{u} - \epsilon, \rp{u} - \epsilon]$ gives another proper
    interval model and it represents a subgraph of $G$.  Since
    $\underline G$ is maximum, we can conclude $|N_{\underline G}[u]|
    \le |N_{\underline G}(v)| < |N_{\underline G}[v]|$.  In this case
    the proof of this claim is concluded.

    In the second case there is no vertex in $K_{\cal A}(\cp{v} +
    \epsilon)$ whose interval intersects $I(u)$ from the left.  Note
    that every interval intersecting $[\lp{w}, \rp{u}]$ represents a
    vertex in $N_G[v]$.  Let $w'$ be vertex in $K_{\cal A}(\cp{v})\cap
    K_{\cal I}(\rp{u})$ that has the leftmost interval, and let $u'$
    be vertex in $K_{\cal A}(\ccp{v})\cap K_{\cal I}(\lp{w})$ that has
    the rightmost interval; these two vertices exist because $u$ and
    $w$ are candidates for them, respectively.  See Fig.~\ref{fig:e}.
    There cannot be any vertex whose interval contains $[\lp{w'},
    \rp{u'}]$: Such a vertex, if it exists, is in $N_G(v)$, but then
    it contradicts the selection of $u'$ and $w'$.  Also by the
    selection of $u'$ and $w'$, no interval contains $[\lp{w'} -
    \epsilon, \rp{u'} + \epsilon]$.  Thus, setting $I(v)$ to $[\lp{w'}
    - \epsilon, \rp{u'} + \epsilon]$ gives another proper interval
    model and it represents a subgraph of $G$.
    Since $\underline G$ is maximum, we can conclude $|N_{\underline
      G}[u]\cap K_{\cal A}(\cp{v} + \epsilon)| \le |N_{\underline
      G}(v)| < |N_{\underline G}[v]| $.
    \renewcommand{\qedsymbol}{$\lrcorner$}    
  \end{proof}

  Therefore, for every vertex $u\in K_{\cal A}(\ccp{v})$, there are no
  less edges incident to $u$ in $E_-$ than in $\overrightarrow E_{\cal
    A}(\ccp{v})$.  Moreover, as there is at least one vertex in
  $K_{\cal A}(\ccp{v})$ adjacent to $K_{\cal A}(\cp{v} + \epsilon)$ in
  $\underline G$, (noting that no edge in $\overrightarrow E_{\cal
    A}(\ccp{v})$ is incident to two vertices in $K_{\cal
    A}(\ccp{v})$,) it follows that $ |E_-| > |\overrightarrow E_{\cal
    A}(\ccp{v})|, $ contradicting that $\underline G$ is a maximum
  unit interval subgraph of $G$.
\end{proof}

It is worth stressing that a thinnest place in an arc model with
respect to edges is not necessarily a thinnest place with respect to
vertices; see Fig.~\ref{fig:thin} for an example.  There is a linear
number of different places to check, and thus the edge deletion
problem can also be solved in linear time on proper Helly circular-arc
graphs.  The problem is also simple on fat $W_5$'s.
\begin{SCfigure}[][h]
  \centering\small
  \includegraphics{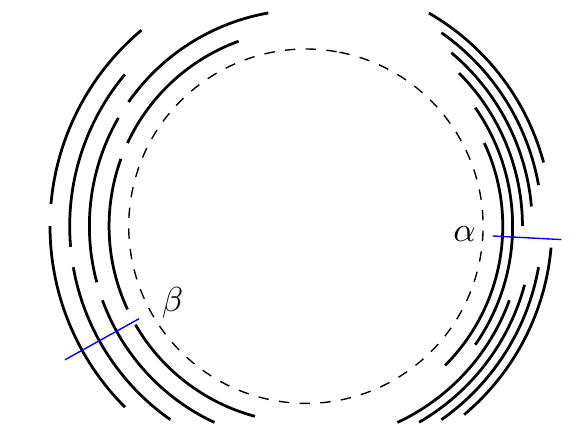}
  \caption{\small The thinnest points for vertices and edges are
    $\alpha$ and $\beta$ respectively:
    \\
    $K_{\cal A}(\alpha) = 2 < K_{\cal A}(\beta) = 3$; while
    \\
    $\protect\overrightarrow E_{\cal A}(\alpha) = 8 >
    \protect\overrightarrow E_{\cal A}(\beta) = 6$.}
  \label{fig:thin}
\end{SCfigure}

\begin{theorem}\label{thm:edge-deletion-phcag}
  The \uied{} problem can be solved in $O(m)$ time (1) on \phcag{s}
  and (2) on $\cal F$-free graphs.
\end{theorem}
\begin{proof}
  For (1), we may assume that the input graph $G$ is not an unit
  interval graph; according to Corollary~\ref{lem:chordal-phcag}, it
  is not chordal.  We build a proper and Helly arc model $\cal A$ for
  $G$; without loss of generality, assume that it is canonical.
  According to Lemma~\ref{lem:edge-hole-cover}, the problem reduces to
  finding a point $\alpha$ in $\cal A$ such that $\overrightarrow
  E_{\cal A}(\alpha)$ is minimized.  It suffices to consider the $2 n$
  points $i + 0.5$ for $i \in \{0, \ldots, 2n - 1\}$.  We calculate
  first $\overrightarrow E_{\cal A}(0.5)$, and then for $i= 1, \ldots,
  2n - 1$, we deduce $\overrightarrow E_{\cal A}(i + 0.5)$ from
  $\overrightarrow E_{\cal A}(i - 0.5)$ as follows.  If $i$ is a
  clockwise endpoint of some arc, then $\overrightarrow E_{\cal A}(i +
  0.5) = \overrightarrow E_{\cal A}(i - 0.5)$.  Otherwise, $i =
  \ccp{v}$ for some vertex $v$, then the difference between
  $\overrightarrow E_{\cal A}(i + 0.5)$ and $\overrightarrow E_{\cal
    A}(i - 0.5)$ is the set of edges incident to $v$.  In particular,
  $\{u v: u \rightarrow v\} = \overrightarrow E_{\cal A}(i -
  0.5)\setminus \overrightarrow E_{\cal A}(i + 0.5)$, while $\{u v: v
  \rightarrow u\} = \overrightarrow E_{\cal A}(i + 0.5)\setminus
  \overrightarrow E_{\cal A}(i - 0.5)$.  Note that the initial value
  $\overrightarrow E_{\cal A}(0.5)$ can be calculated in $O(m)$ time,
  and then each vertex and its adjacency list is scanned exactly once.
  It follows that the total running time is $O(m)$.

  For (2), we may assume that the input graph $G$ is connected, as
  otherwise we work on its components one by one.  According to
  Theorem~\ref{thm:characterization}(1), $G$ is either a proper Helly
  circular-arc graph or a fat $W_5$.  The former case has been
  considered above, and now assume $G$ is a fat $W_5$.  Let $K_0,
  \ldots, K_4$ be the five cliques in the fat hole, and let $K_5$ be
  the hub.  
  We may look for a maximum spanning unit interval subgraph
  $\underline G$ of $G$ such that $N_{\underline G}[u] = N_{\underline
    G}[v]$ for any pair of vertices $u,v$ in $K_i$, where $i \in \{0,
  \ldots, 5\}$.  We now argue the existence of such a subgraph.  By
  definition, $N_G[u] = N_G[v]$.  Let ${G'}$ be a maximum spanning
  unit interval subgraph of $G$ where $N_{G'}[u] \ne N_{G'}[v]$, and
  assume without loss of generality, $|N_{G'}[u]| \ge |N_{G'}[v]|$.
  We may change the deleted edges that are incident to $v$ to make
  another subgraph where $v$ has the same neighbors as $N_{G'}[u]$;
  this graph is clearly a unit interval graph and has no less edges
  than $G'$.  This operation can be applied to any pair of $u, v$ in
  the same twin class, and it will not violate an earlier pair.
  Repeating it we will finally end with a desired maximum spanning
  unit interval subgraph.  Therefore, there is always some $i \in \{0,
  \ldots, 4\}$ (all subscripts are modulo $5$) such that deleting all
  edges between cliques ${K_i}$ and ${K_{i+1}}$ together with edges
  between one of them and ${K_5}$ leaves a maximum spanning unit
  interval subgraph.  Once the sizes of all six cliques have been
  calculated, which can be done in $O(m)$ time, the minimum set of
  edges can be decided in constant time.  Therefore, the total running
  time is $O(m)$.  The proof is now complete.
\end{proof}
Indeed, it is not hard to see that in the proof of
Theorems~\ref{thm:edge-deletion-phcag}(2), \emph{every} maximum
spanning unit interval subgraph of a fat $W_5$ keeps its six twin
classes, but we are satisfied with the weaker statement that is
sufficient for our algorithm.  

Theorems~\ref{thm:edge-deletion-phcag} and \ref{thm:characterization}
already imply a branching algorithm for the \uied{} problem running in
time $O(9^k\cdot m)$.  Here the constant $9$ is decided by the $S_3$,
which has $9$ edges.  However, a closer look at it tells us that
deleting any single edge from an $S_3$ introduces either a claw or a
$C_4$, which forces us to delete some other edge(s).  The disposal of
an $\overline{S_3}$ is similar.  The labels for an $S_3$ and a
$\overline{S_3}$ used in the following proof are as given in
Fig.~\ref{fig:fis}.
\begin{proposition}\label{lem:4-suffice}
  Let $\underline G$ be a spanning unit interval subgraph of a graph
  $G$, and let $F = E(G)\setminus E(\underline G)$. 
  \begin{enumerate}[(1)]
  \item For an $S_3$ of $G$, there must be some $i = 1,2,3$ such that
    $F$ contains at least two edges from the triangle involving $u_i$.
  \item For an $\overline S_3$ of $G$, the set $F$ contains either an
    edge $u_i v_i$ for some $i = 1,2,3$, or at least two edges from
    the triangle $u_1 u_2 u_3$.
  \end{enumerate}
\end{proposition}
\begin{proof}
  (1) We consider the intervals for $v_1, v_2$, and $v_3$ in a unit
  interval model $\cal I$ for $\underline G$.  If $v_1 v_2 v_3$
  remains a triangle of $\underline G$, then the interval
  $I({v_1})\cup I({v_2})\cup I({v_3})$ has length less than $2$, and
  it has to be disjoint from $I(u_i)$ for at least one $i\in \{1, 2,
  3\}$.  In other words, both edges incident to this $u_i$ are in $F$.
  Otherwise, $v_1 v_2 v_3$ is not a triangle of $\underline G$.
  Assume without loss of generality $\lp{v_1} < \lp{v_2} < \lp{v_3}$.
  Then $v_1 v_3\not\in E(\underline G)$ and $u_2$ cannot be adjacent
  to both $v_1$ and $v_3$.  Therefore, at least two edges from the
  triangle $u_2 v_1 v_3$ are in $F$; other cases are symmetric.

  (2) If $F$ contains none of the three edges $v_1 u_1$, $v_2 u_2$,
  and $v_3 u_3$, then it contains at least two edges from the triangle
  $u_1 u_2 u_3$: Otherwise there is a claw.
\end{proof}

This observation and a refined analysis will yield the running time
claimed in Theorem~\ref{thm:alg-1}.  The algorithm goes similarly as
the parameterized algorithm for \uivd{} used in the proof of
Theorem~\ref{thm:algs-uivd}.
\begin{theorem}
  The \uied{} problem can be solved in time $O(4^k\cdot m)$.
\end{theorem}
\begin{proof}
  The algorithm calls first Theorem~\ref{thm:characterization}(2) to
  decide whether there exists an induced subgraph in $\cal F$, and
  then based on the outcome, it solves the problem by making recursive
  calls to itself, or calling the algorithm of
  Theorem~\ref{thm:edge-deletion-phcag}.  When a claw or $C_4$ is
  found, the algorithm makes respectively $3$ or $4$ calls to itself,
  each with a new instance with parameter value $k - 1$ (deleting one
  edge from the claw or $C_4$).  For an $S_3$, the algorithm branches
  on deleting two edges from a triangle involving a vertex $u_i$ with
  $i = 1, 2, 3$.  Since there are three such triangles, and each has
  three options, the algorithm makes $9$ calls to itself, all with
  parameter value $k - 2$.  For an $\overline{S_3}$, the algorithm
  makes $6$ calls to itself, of which $3$ with parameter value $k - 1$
  (deleting edge $v_i u_i$ for $i = 1, 2, 3$), and another $3$ with
  parameter value $k - 2$ (deleting two edges from the triangle $v_1
  v_2 v_3$).

  To verify the correctness of the algorithm, it suffices to show that
  for any spanning unit interval subgraph $\underline G$ of $G$, there
  is at least one recursive call that generates a graph $G'$
  satisfying $E(\underline G)\subseteq E(G')\subseteq E(G)$.  This is
  obvious when the recursive calls are made on a claws or $C_4$.  It
  follows from Proposition~\ref{lem:4-suffice} when the recursive
  calls are made on an $S_3$ or $\overline{S_3}$.  With standard
  technique, it is easy to verify that $O(4^k)$ recursive calls are
  made, each in $O(m)$ time.  Moreover, the algorithm for
  Theorem~\ref{thm:edge-deletion-phcag} is called $O(4^k)$ times.  It
  follows that the total running time of the algorithm is $O(4^k\cdot
  m)$.
\end{proof}

What dominates the branching step is the disposal of $C_4$'s.  With
the technique the author developed in \cite{cao-15-edge-deletion}, one
may (slightly) improve the running time to $O(c^k\cdot m)$ for some
constant $c < 4$.  To avoid blurring the focus of the present paper,
we omit the details.

\section{General editing}\label{sec:general}
Let $V_-\subseteq V(G)$, and let $E_-$ and $E_+$ be a set of edges and
a set of non-edges of $G - V_-$ respectively.  We say that ($V_-, E_-,
E_+$) is an \emph{editing set} of $G$ if the deletion of $E_-$ from
and the addition of $E_+$ to $G - V_-$ create a unit interval graph.
Its {\em size} is defined to be the 3-tuple ($|V_-|, |E_-|, |E_+|$),
and we say that it is {\em smaller} than ($k_1,k_2,k_3$) if all of
$|V_-|\le k_1$ and $|E_-|\le k_2$ and $|E_+|\le k_3$ hold true and at
least one inequality is strict.  The \uie{} problem is formally
defined as follows.

\medskip
\fbox{\parbox{0.9\linewidth}{
\begin{tabularx}{\linewidth}{rX}
  \textit{Input:} & A graph $G$ and three nonnegative integers $k_1$,
  $k_2$, and $k_3$.
  \\
  \textit{Task:} & Either construct an editing set $(V_-,E_-,E_+)$ of
  $G$ that has size at most ($k_1,k_2,k_3$), or report that no such
  set exists.
\end{tabularx}
}}
\medskip
\\
We remark that it is necessary to impose the quotas for different
modifications in the stated, though cumbersome, way.  Since vertex
deletions are clearly preferable to both edge operations, the problem
would be computationally equivalent to \uivd{} if we have a single
budget on the total number of operations.

By and large, our algorithm for the \uie{} problem also uses the same
two-phase approach as the previous algorithms.  The main discrepancy
lies in the first phase, when we are not satisfied with a \phcag{} or
an $\cal F$-free graph.  In particular, we also want to dispose of all
holes $C_\ell$ with $\ell \le k_3 + 3$, which are precisely those
holes fixable by merely adding edges (recall that at least $\ell - 3$
edges are needed to fill a $C_\ell$ in).  In the very special cases
where $k_3 = 0$ or $1$, a fat $W_5$ is ${\cal F}\cup \{C_\ell : \ell
\le k_3 + 3\}$-free.  It is not hard to solve fat $W_5$'s, but to make
the rest more focused and also simplify the presentation, we also
exclude these cases by disposing of all $C_5$'s in the first phase.

A graph is called \emph{reduced} if it contains no claw, $S_3,
\overline{S_3}$, $C_4$, $C_5$, or $C_\ell$ with $\ell \le k_3 + 3$.
By Proposition~\ref{lem:phcag}, a reduced graph $G$ is a \phcag{}.
Hence, if $G$ happens to be chordal, then it must be a unit interval
graph (Corollary~\ref{lem:chordal-phcag}), and we terminate the
algorithm.  Otherwise, our algorithm enters the second phase.  Now
that $G$ is reduced, every minimal forbidden induced subgraph is a
hole $C_\ell$ with $\ell > k_3 + 3$, which can only be fixed by
deleting vertices and/or edges.  Here we again exploit a proper and
Helly arc model $\cal A$ for $G$.  According to
Lemma~\ref{lem:hole-cover}, if there exists some point $\rho$ in the
model such that $|K_{\cal A}(\rho)| \le k_1$, then it suffices to
delete all vertices in $K_{\cal A}(\rho)$, which results in a subgraph
that is a unit interval graph.  Therefore, we may assume hereafter
that no such point exists, then $G$ remains reduced and non-chordal
after at most $k_1$ vertex deletions.  As a result, we have to delete
edges as well.

Consider an (inclusion-wise minimal) editing set ($V_-, E_-, E_+$) to
a reduced graph $G$.  It is easy to verify that ($\emptyset, E_-,
E_+$) is an (inclusion-wise minimal) editing set of the reduced graph
$G - V_-$.  In particular, $E_-$ needs to intersect all holes of $G -
V_-$.  We use ${\cal A} - V_-$ as a shorthand for $\{A(v)\in {\cal A}:
v\not\in V_-\}$, an arc model for $G - V_-$ that is proper and Helly.
One may want to use Lemma~\ref{lem:edge-hole-cover} to find a minimum
set $E_-$ of edges (i.e., $\overrightarrow E_{{\cal A} - V_-}(\alpha)$
for some point $\alpha$) to finish the task.  However,
Lemma~\ref{lem:edge-hole-cover} has not ruled out the possibility that
we delete less edges to break all long holes, and subsequently add
edges to fix the incurred subgraphs in \{claw, $\overline{S_3}$,
$S_3$, $C_4$, $C_5$, $C_\ell$\} with $\ell \le k_3 + 3$.  So we need
the following lemma.
\begin{lemma}\label{lem:no-addition}
  Let ($V_-, E_-, E_+$) be an inclusion-wise minimal editing set of a
  reduced graph $G$.  If $|E_+| \le k_3$, then $E_+ = \emptyset$.
\end{lemma}
\begin{proof}
  We may assume without loss of generality $V_- = \emptyset$, as
  otherwise it suffices to consider the inclusion-wise minimal editing
  set ($\emptyset, E_-, E_+$) to the still reduced graph $G - V_-$.
  Let $\cal A$ be a proper and Helly arc model for $G$.  Let $E'_-$ be
  an inclusion-wise minimal subset of $E_-$ such that for every hole
  in $G - E'_-$, the union of arcs for its vertices does not cover the
  circle of $\cal A$.  We argue the existence of $E'_-$ by showing
  that $E_-$ itself satisfies this condition.  Suppose for
  contradiction that there exists in $G - E_-$ a hole whose arcs cover
  the circle of $\cal A$.  Then we can find a minimal subset of them
  that covers the circle of $\cal A$.  By
  Corollary~\ref{cor:non-chordal-model}, this subset has at least $k_3
  + 4$ vertices, and thus the length of the hole in $G - E_-$ is at
  least $k_3 + 4$.  But then it cannot be fixed by the addition of the
  at most $k_3$ edges from $E_+$.

  Now for the harder part, we argue that $\underline G := G - E'_-$ is
  already a unit interval graph.  Together with the inclusion-wise
  minimality, it would imply $E_- = E'_-$ and $E_+ = \emptyset$.

  Suppose for contradiction that $\underline G[X]$ is a claw, $S_3,
  \overline{S_3}$, or a hole for some $X\subseteq V(G)$.  We find
  three vertices $u, v, w\in X$ such that $u w\in E'_-$ and $u v, v
  w\in E(\underline G)$ as follows.  By
  Corollary~\ref{cor:non-chordal-model} and the fact that $G$ is
  $\{C_4, C_5\}$-free, at least six arcs are required to cover the
  circle.  As a result, if the arcs for a set $Y$ of at most six
  vertices covers the circle, then $G[Y]$ must be a $C_6$, and its
  subgraph $\underline G[Y]$ cannot be a claw, $\overline{S_3}$, or
  $S_3$.  Therefore, $\bigcup_{v\in X} A(v)$ cannot cover the whole
  circle when $\underline G[X]$ is a claw, $S_3$, or $\overline{S_3}$.
  On the other hand, from the selection of $E'_-$, this is also true
  when $\underline G[X]$ is a hole.  Thus, $G[X]$ is a unit interval
  graph, and we can find two vertices $x, z$ from $X$ having $x z \in
  E'_-$.  We find a shortest $x$-$z$ path in $\underline G[X]$.  If
  the path has more than one inner vertex, then it makes a hole
  together with $x z$; {as $G[X]$ is a unit interval graph}, this
  would imply that there exists an inner vertex $y$ of this path such
  that $x y\in E'_-$ or $y z\in E'_-$.  We consider then the new pair
  $x, y$ or $y, z$ accordingly.  Note that their distance in
  $\underline G[X]$ is smaller than $x z$, and hence repeating this
  argument (at most $|X| - 3$ times) will end with two vertices with
  distance precisely $2$ in $\underline G[X]$.  They are the desired
  $u$ and $w$, while any common neighbor of them in $\underline G[X]$
  can be $v$.

  By the minimality of $E'_-$, in $\underline G + u w$ there exists a
  hole $H$ such that arcs for its vertices cover the circle in $\cal
  A$.  This hole $H$ necessarily passes $u w$, and we denote it by
  $x_1 x_2 \cdots x_{\ell -1} x_\ell$, where $x_1 = u$ and $x_\ell =
  w$.  Note that $A(u)$ intersects $A(w)$, and since $\cal A$ is
  proper and Helly, $A(u), A(v), A(w)$ cannot cover the circle;
  moreover, it cannot happen that $A(v)$ intersects all the arcs
  $A(x_i)$ for $1<i<\ell$ simultaneously.  From $x_1 x_2 \cdots
  x_{\ell -1} x_\ell$ we can find $p$ and $q$ such that $1\le p < p+1
  < q\le \ell$ and $v x_p, v x_q\in E(\underline G)$ but $v x_i\not\in
  E(\underline G)$ for every $p < i < q$.  Here possibly $p = 1$
  and/or $q = \ell$.  Then $v x_p \cdots x_q$ makes a hole of
  $\underline G$, and the union of its arcs covers the circle,
  contradicting the definition of $E'_-$.  This concludes the proof.
\end{proof}

Therefore, a yes-instance on a reduced graph always has a solution
that does not add any edge.  By Lemma~\ref{lem:edge-hole-cover}, for
any editing set ($V_-, E_-, \emptyset$), we can always find some point
$\alpha$ in the model and use $\overrightarrow E_{{\cal A} -
  V_-}(\alpha)$ to replace $E_-$.  After that, we can use the vertices
``close'' to this point to replace $V_-$.  Therefore, the problem
again boils down to find some ``weak point'' in the arc model.  This
observation is formalized in the following lemma.  We point out that
this result is stronger than required by the linear-time algorithm,
and we present in the current form for its own interest (see
Section~\ref{sec:remarks} for more discussions).
\begin{lemma}\label{lem:mixed-hole-covers}
  Given a \phcag{} $G$ and a nonnegative integer $p$, we can calculate
  in $O(m)$ time the minimum number $q$ such that $G$ has an editing
  set of size ($p, q, 0$).  In the same time we can find such an
  editing set.
\end{lemma}
\begin{proof}
  We may assume that $G$ is not chordal; otherwise, by
  Corollary~\ref{lem:chordal-phcag}, $G$ is a unit interval graph and
  the problem becomes trivial because an empty set will suffice.  Let
  us fix a proper and Helly arc model $\cal A$ for $G$.  The lemma
  follows from Lemma~\ref{lem:hole-cover} when there is some point
  $\rho$ satisfying $K_{\cal A}(\rho) \le p$.  Hence, we may assume
  that no such point exists, and for any subset $V_-$ of at most $p$
  vertices, $G - V_-$ remains a \phcag{} and is non-chordal.  Hence,
  $q > 0$.  For each point $\rho$ in $\cal A$, we can define an
  editing set ($V^\rho_-, E^\rho_-, \emptyset$) by taking the $p$
  vertices in $K_{\cal A}(\rho)$ with the most clockwise arcs as
  $V^\rho_-$ and $\overrightarrow E_{{\cal A} - V^\rho_-}(\rho)$ as
  $E^\rho_-$.  We argue first that the minimum cardinality of this
  edge set, taken among all points in $\cal A$ is the desired number
  $q$.  See Fig.~\ref{fig:mixed-hole-covers}.

  Let ($V^*_-, E^*_-, \emptyset$) be an editing set of $G$ with size
  ($p, q, 0$).  According to Lemma~\ref{lem:edge-hole-cover}, there is
  a point $\alpha$ such that the deletion of $E'_- := \overrightarrow
  E_{{\cal A} - V^*_-}(\alpha)$ from $G - V^*_-$ makes it a unit
  interval graph and $|E'_-| \le |E^*_-|$.  We now consider the
  original model $\cal A$.  Note that a vertex in $V^*_-$ is in either
  $K_{\cal A}(\alpha)$ or $\{v\not\in K_{\cal A}(\alpha): u\rightarrow
  v, u\in K_{\cal A}(\alpha)\}$; otherwise replacing this vertex by
  any end of an edge in $E^*_-$, and removing this edge from $E^*_-$
  gives an editing set of size ($p, q - 1, 0$).  Let $V_-$ comprise
  the $|V^*_-\cap K_{\cal A}(\alpha)|$ vertices of $K_{\cal
    A}(\alpha)$ whose arcs are the most clockwise in them, as well as
  the first $|V^*_-\setminus K_{\cal A}(\alpha)|$ vertices whose arcs
  are immediately to the right of $\alpha$.  And let $E_- :=
  \overrightarrow E_{{\cal A}- V_-}(\alpha)$.  It is easy to verify
  that $|E_-| \le |E^*_-| = q$ and ($V_-, E_-, \emptyset$) is also an
  editing set of $G$ (Lemma~\ref{lem:edge-hole-cover-1}).  Note that
  arcs for $V_-$ are consecutive in $\cal A$.  Let $v$ be the vertex
  in $V_-$ with the most clockwise arc, and then $\ccp{v} + \epsilon$
  is the desired point $\rho$.

  We give now the $O(m)$-time algorithm for finding the desired point,
  for which we assume that $\cal A$ is canonical.  It suffices to
  consider the $2 n$ points $i + 0.5$ for $i \in \{0, \ldots, 2n -
  1\}$.  We calculate first the $V^{0.5}_-$ and $E^{0.5}_-$, and
  maintain a queue that is initially set to be $V^{0.5}_-$.  For $i =
  1, \ldots, 2n - 1$, we deduce the new sets $V^{i + 0.5}_-$ and $E^{i
    + 0.5}_-$ from $V^{i - 0.5}_-$ and $E^{i - 0.5}_-$ as follows.  If
  $i$ is a clockwise endpoint of some arc, then $V^{i + 0.5}_- = V^{i
    - 0.5}_-$ and $E^{i + 0.5}_- = E^{i - 0.5}_-$.  Otherwise, $i =
  \ccp{v}$ for some vertex $v$, then we enqueue $v$, and dequeue $u$.
  We set $V^{i + 0.5}_-$ to be the vertices in the queue, whose size
  remains $p$.  The different edges between $E^{i + 0.5}_-$ and $E^{i
    - 0.5}_-$ are those incident to $u$ and $v$.  In particular, $E^{i
    + 0.5}_- \setminus E^{i - 0.5}_- = \{u x: u \rightarrow x,
  x\not\in V^{i + 0.5}_-\}$, while $E^{i - 0.5}_- \setminus E^{i +
    0.5}_- = \{x v: x \rightarrow v, x\not\in V^{i - 0.5}_-\}$.  Note
  that the initial sets $V^{0.5}_-$ and $E^{0.5}_-$ can be found in
  $O(m)$ time, and then each vertex and its adjacency is scanned
  exactly once.  The total running time is $O(m)$.  This concludes the
  proof.
\end{proof}
\begin{SCfigure}[][h]
  \centering\small
    \includegraphics{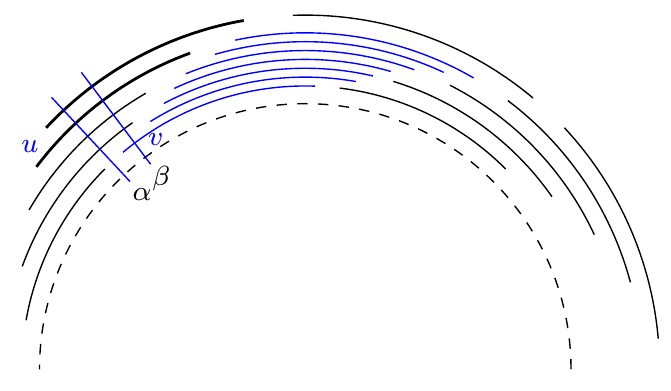} 
  \caption{Illustration for the proof of
    Lemma~\ref{lem:mixed-hole-covers}.  Here $p = 2$ and $q = 3$.
    Then $V^\alpha_-$ consists of the two thick arcs, and
    $|E^\alpha_-| = 3$.  Moving from point $\alpha$ to $\beta$ ($=
    \alpha + 1$) gives $|E^\beta_-| = |E^\alpha_-| + 4 - 2 = 5$.}
  \label{fig:mixed-hole-covers}
\end{SCfigure}

  Again, one should note that in the general case (when both $p, q >
  0$), the point identified by Lemma~\ref{lem:mixed-hole-covers} may
  not be the thinnest point for vertices or the thinnest point for
  edges, as specified by respectively Lemmas~\ref{lem:hole-cover} and
  \ref{lem:edge-hole-cover}.  Indeed, for different values of $p$, the
  thinnest points found by Lemma~\ref{lem:mixed-hole-covers} may be
  different.

  The mixed hole covers consists of both vertices and edges, and thus
  the combinatorial characterization given in
  Lemma~\ref{lem:mixed-hole-covers} extends
  Lemmas~\ref{lem:hole-cover} and \ref{lem:edge-hole-cover}.  The
  algorithm used in the proof is similar as that of
  Theorem~\ref{thm:edge-deletion-phcag}.  Recall that a reduced graph
  is a \phcag{}.  Thus, Lemmas~\ref{lem:no-addition} and
  \ref{lem:mixed-hole-covers} have the following consequence: It
  suffices to call the algorithm with $p = k_1$, and returns the found
  editing set if $q \le k_2$, or ``NO'' otherwise.
\begin{corollary}\label{cor:editing-reduced}
  The \uie{} problem can be solved in $O(m)$ time on reduced graphs.
\end{corollary}

Putting together these steps, the fixed-parameter tractability of
\uie{} follows.  Note that to fill a hole, we need to add an edge
whose ends have distance $2$.
\begin{proof}[Proof of Theorem~\ref{thm:alg-interval-editing}]
  We start by calling Theorem~\ref{thm:characterization}.  If a
  subgraph in $\cal F$ or $W_5$ is detected, then we branch on all
  possible ways of destroying it or the contained $C_5$.  Otherwise,
  we have in our disposal a proper and Helly arc model for $G$, and we
  call Lemma~\ref{lem:shortest-hole} to find a shortest hole $C_\ell$.
  If $\ell \le k_3 + 3$, then we either delete one of its $\ell$
  vertices and $\ell$ edges, or add one of $\ell$ edges $h_i h_{i+2}$
  (the subscripts are modulo $\ell$).  One of the three parameters
  decreases by $1$.  We repeat these two steps until some parameter
  becomes negative, then we terminate the algorithm by returning
  ``NO''; or the graph is reduced, and then call the algorithm of
  Corollary~\ref{cor:editing-reduced} to solve it.  The correctness of
  this algorithm follows from Lemma~\ref{lem:shortest-hole} and
  Corollary~\ref{cor:editing-reduced}.  In the disposal of a subgraph
  of $\cal F$, at most $21$ recursive calls are made, while $3\ell$
  for a $C_\ell$, each having a parameter $k - 1$.  Therefore, the
  total number of instances (with reduced graphs) made in the
  algorithm is $O((3 k_3 + 21)^k)$.  It follows that the total running
  time of the algorithm is $2^{O(k\log k)}\cdot m$.
\end{proof}

It is worth mentioning that Lemma~\ref{lem:mixed-hole-covers} actually
implies a linear-time algorithm for the unit interval deletion problem
(which allows $k_1$ vertex deletions and $k_2$ edge deletions) on the
\phcag{s} and an $O(10^{k_1 + k_2} \cdot m)$-time algorithm for it on
general graphs.  The constant 10 can be even smaller if we notice that
(1) the problem is also easy on fat $W_5$'s, and (2) the worst cases
for vertex deletions ($S_3$'s and $\overline{S_3}$'s) and edge
deletions ($C_4$'s) are different.

\section{Concluding remarks}\label{sec:remarks}

All aforementioned algorithms exploit the characterization of unit
interval graphs by forbidden induced subgraphs
\cite{wegner-67-dissertation}.  Very recently, Bliznets et
al.~\cite{bliznets-14-unit-interval-completion} used a different
approach to produce a subexponential-time parameterized algorithm for
\uic{} (whose polynomial factor is however not linear).  Using a
parameter-preserving reduction from vertex cover
\cite{lewis-80-node-deletion-np}, one can show that the vertex
deletion version cannot be solved in $2^{o(k)}\cdot n^{O(1)}$ time,
unless the Exponent Time Hypothesis fails
\cite{cai-96-hereditary-graph-modification}.  Now that the edge
deletion version is FPT as well, one may want to ask to which side it
belongs.  The evidence we now have is in favor of the hard side: In
all related graph classes, the edge deletion versions seem to be
harder than their vertex deletion counterparts.

As said, it is not hard to slightly improve the constant $c$ in the
running time $O(c^k\cdot m)$, but a significant improvement would need
some new observation(s).  More interesting would be to fathom their
limits.  In particular, can the deletion problems be solved in time
$O(2^k\cdot m)$?

Polynomial kernels for \uic{}
\cite{bessy-13-kernels-proper-interval-completion} and \uivd{}
\cite{fomin-13-kernel-pivd} have been known for a while.  Using the
approximation algorithm of Theorem~\ref{thm:alg-2},
we~\cite{cao-16-kernel-uivd} recently developed an $O(k^4)$-vertex
kernel for \uivd{}, improving from the $O(k^{53})$ one of Fomin et
al.~\cite{fomin-13-kernel-pivd}.  We conjecture that the \uied{}
problem also has a small polynomial kernel.

The algorithm for \uie{} is the second nontrivial FPT algorithm for
the general editing problem.  The main ingredient of our algorithm is
the characterization of the mixed deletion of vertices and edges to
break holes.  A similar study has been conducted in the algorithm for
the chordal editing problem \cite{cao-16-chordal-editing}.  In
contrast to that, Lemmas~\ref{lem:no-addition} and
\ref{lem:mixed-hole-covers} are somewhat stronger.  For example, we
have shown that once small forbidden subgraphs have been all fixed, no
edge additions are further needed.  Together with Marx, we had
conjectured that this is also true for the chordal editing problem,
but we failed to find a proof.  Very little study had been done on the
mixed deletion of vertices and edges
\cite{marx-13-treewidth-reduction}.  We hope that our work will
trigger more studies on this direction, which will further deepen our
understanding of various graph classes.

We point out that although we start by breaking small forbidden
induced subgraphs, our major proof technique is instead manipulating
(proper/unit) interval models.  The technique of combining
(constructive) interval models and (destructive) forbidden induced
subgraphs is worth further study on related problems.

\section*{Appendix}
\begin{figure*}[t]
  \centering\small
  \subfloat[$F_1$]{\label{fig:f1}
    \includegraphics{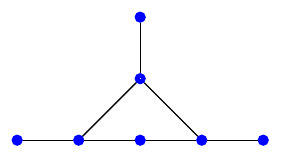} 
  }
  \subfloat[$F_2$]{\label{fig:f2}
    \includegraphics{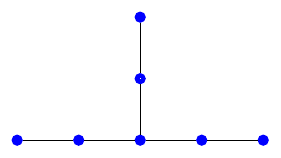} 
  }
  \subfloat[$F_3$]{\label{fig:whipping-top}
    \includegraphics{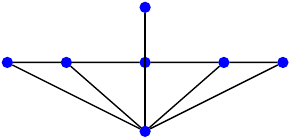} 
  }
  \subfloat[\dag]{\label{fig:net}
    \includegraphics{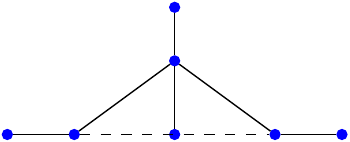} 
  }
  \subfloat[\ddag]{\label{fig:tent}
    \includegraphics{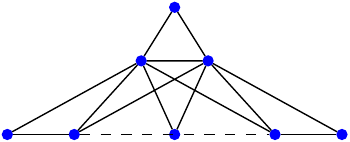} 
  }

  \subfloat[$K_{2,3}$]{\label{fig:long-claw-1}
    \includegraphics{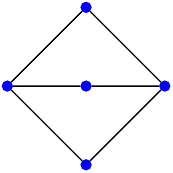} 
  }
  $\quad$
  \subfloat[$F_4$]{\label{fig:g-2}
    \includegraphics{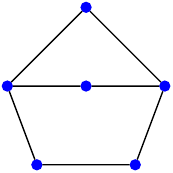} 
  }
  $\quad$
  \subfloat[$F_5$]{\label{fig:domino}
    \includegraphics{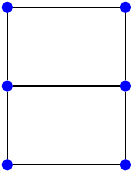} 
  }
  $\quad$
  \subfloat[$F_6$]{\label{fig:f3}
    \includegraphics{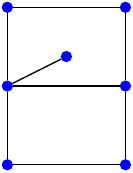} 
  }
  $\quad$
  \subfloat[$F_7$]{\label{fig:f4}
    \includegraphics{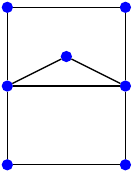} 
  }
  $\quad$
  \subfloat[$F_8$]{\label{fig:fis-1}
    \includegraphics{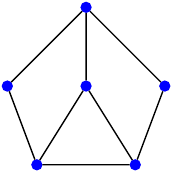} 
  }
  $\quad$
  \subfloat[$F_9$]{\label{fig:fis-2}
    \includegraphics{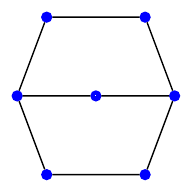} 
  }
  \caption{Forbidden induced graphs.}
  \label{fig:small-graphs}
\end{figure*}

\begin{figure}[t]
  \centering\small
  \captionsetup{justification=centering}
  \includegraphics{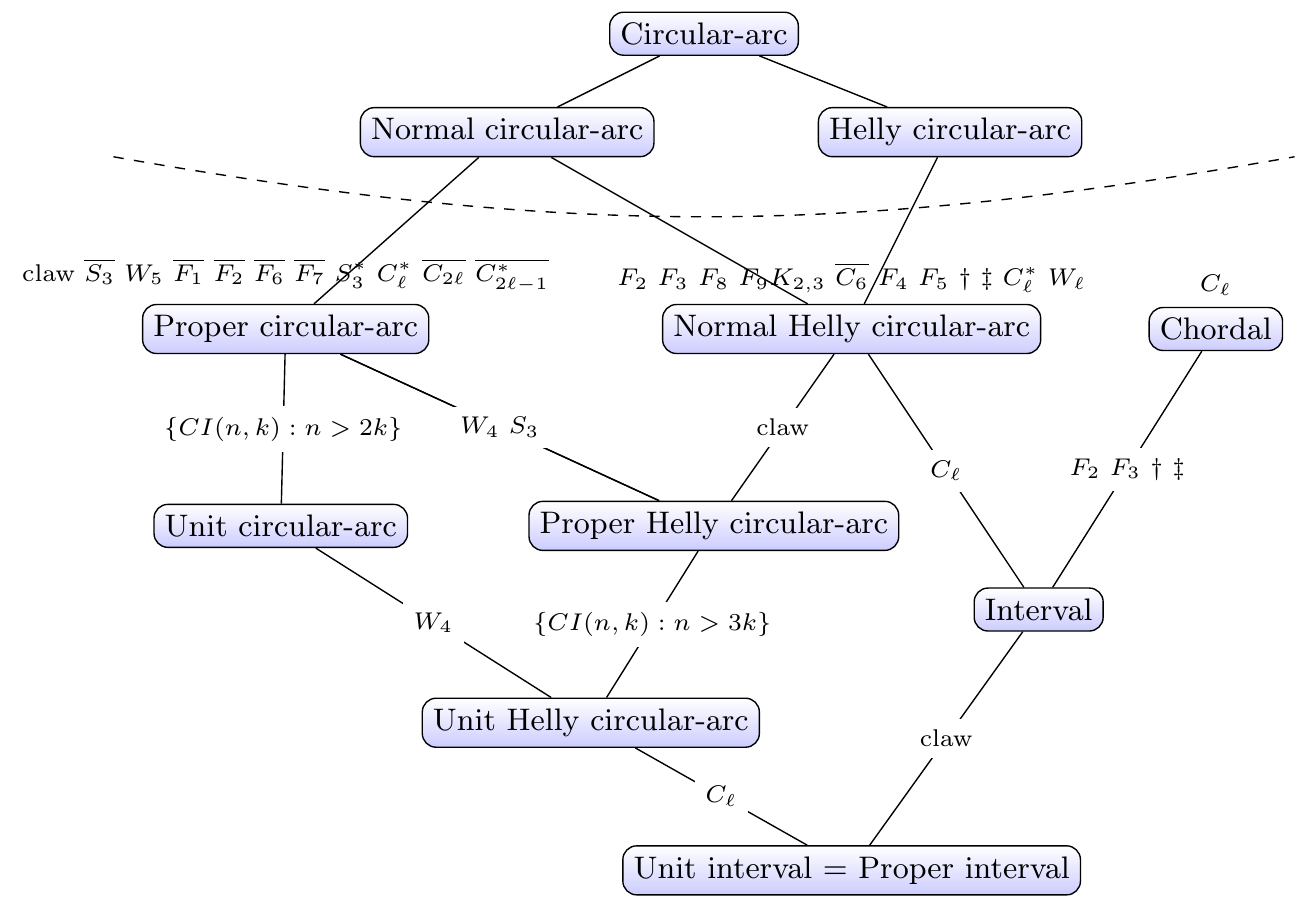} 
  \caption{Forbidden induced subgraphs and containment relations of
    related graph classes ($\ell\ge 4$)}
  \label{fig:classes-containment}
\end{figure}

For the convenience of the reader, we collect related graph classes
and their containment relations in Fig.~\ref{fig:classes-containment},
which is adapted from Lin et al.~\cite{lin-13-nhcag-and-subclasses}
(note that some of these graph classes are not used in the present
paper).  The $CI(n, k)$ graphs are defined by
Tucker~\cite{tucker-74-structures-cag}; see also
\cite{lin-13-nhcag-and-subclasses}.  Other subgraphs that have not
been introduced in the main text are depicted in
Fig.~\ref{fig:small-graphs}.  The relations in
Fig.~\ref{fig:classes-containment} can be viewed both from the
intersection models, arcs or intervals, and forbidden induced
subgraphs, every minimal forbidden induced subgraph of a super-class
being a (not necessarily minimal) forbidden induced subgraph of its
subclass.  For example, Proposition~\ref{lem:non-chordal-model} and
Corollary~\ref{cor:non-chordal-model} are actually properties of
\nhcag{s}.  For a \nhcag{} that is not chordal, every arc model has to
be normal and Helly \cite{lin-13-nhcag-and-subclasses, cao-15-nhcag}.
This is also true for the subclass of \phcag{s}, but an arc model for
a \phcag{} may not be proper.

A word of caution is worth on the definition of \phcag{s}.  One graph
might admit two arc models, one being proper and the other Helly, but
no arc model that is both proper and Helly, e.g., the $S_3$ and the
$W_4$.  Therefore, the class of proper Helly circular-arc graphs does
not contain all those graphs being both proper circular-arc graphs and
Helly circular-arc graphs, but a proper subclass of it.  A similar
remark applies to \nhcag{s}.

For the three classes at the top of
Fig.~\ref{fig:classes-containment}, their characterizations by minimal
forbidden induced subgraphs are still open.  At the third level, the
minimal forbidden induced subgraphs for proper circular-arc graphs and
\nhcag{s} are completely determined by
Tucker~\cite{tucker-74-structures-cag} and Cao et
al.~\cite{cao-15-nhcag}.  For all the classes at lower levels, their
forbidden induced subgraphs with respect to its immediate
super-classes are given.  From them we are able to derive all the
minimal forbidden induced subgraphs for each of these classes.

For example, the characterization of unit interval graphs
(Theorem~\ref{thm:uig-fis}) follows from the characterization of
interval graphs and that we can find a claw in an $F_2$, an $F_3$, a
$\dag$ that is not an $\overline{S_3}$, or a $\ddag$ that is not an
${S_3}$.  Likewise, the minimal forbidden induced subgraphs of proper
Helly circular-arc graphs stated in Theorem~\ref{thm:phcag} can be
derived from those of proper circular-arc graphs and by Corollary~5 of
\cite{lin-13-nhcag-and-subclasses}: A proper circular-arc graph that
is not a proper Helly circular-arc graph must contain a $W_4$ or
$S_3$.  Clearly, $S^*_3$ contains an $S_3$.  To see that each of
$\overline{F_1}, \overline{F_2}, \overline{F_6}, \overline{F_7}$, and
$\overline{C_{2 \ell}}, \overline{C_{2 \ell - 1}^*}$ for $\ell \ge 4$
contain a ${W_4}$, it is equivalent to check that each of ${F_1},
{F_2}, {F_6}, {F_7}$, and ${C_{2 \ell}}, {C_{2 \ell - 1}^*}$ for $\ell
\ge 4$ contains a $\overline{W_4}$, i.e., two non-incident edges and
another independent vertex $v$.  This can be directly read from
Fig.~\ref{fig:small-graphs} for ${F_1}, {F_2}, {F_6}, {F_7}$.  Let
$h_1, h_2,\ldots$ denote the vertices in the hole of ${C_{2 \ell}}$
and ${C_{2 \ell - 1}^*}$.  Then edges $h_1 h_2$ and $h_4 h_5$ are
non-incident.  In a $C_7^*$, the vertex not in the hole can be the
$v$, while in all other holes longer than $7$, the vertex $h_7$ can be
the $v$.

{
  \small
  \bibliographystyle{plainurl}
  \bibliography{../journal,../main}
}
\end{document}